\documentclass{llncs}

\usepackage{amsmath, amssymb}
\usepackage{xspace}
\usepackage{dsfont}
\usepackage[colorlinks=true,urlcolor=webblue,linkcolor=webgreen,filecolor=webblue,citecolor=webgreen,pdfpagemode=UseOutlines,pdfstartview=FitH,pdfpagelayout=OneColumn,bookmarks=true]{hyperref}
\usepackage{fullpage}
\usepackage{enumerate}
\usepackage{graphicx}
\usepackage{algorithm, algpseudocode, float}

\usepackage{tikz}

\definecolor{webgreen}{rgb}{0,.5,0}
\definecolor{webblue}{rgb}{0,0,.5}

\usepackage{ntheorem}

\theoremstyle{plain}
\theorembodyfont{}
\theoremsymbol{}
\theoremprework{}
\theorempostwork{}
\theoremseparator{.}

\newtheorem{experiment}{Experiment}
\newtheorem{scheme}{Scheme}

\newtheorem{condition}{Condition}

\newcommand{\N}{\mathbb{N}}
\newcommand{\NN}{\mathbb{N}}

\newcommand{\foral}{\ensuremath{\forall \, }\xspace}

\newcommand{\id}{\mathrm{id}}

\newcommand{\ket}[1]{| #1 \rangle}
\newcommand{\bra}[1]{\langle #1 |}

\newcommand{\ketbra}[2]{\left|#1\right\rangle\!\!\left\langle #2\right|}

\newcommand{\tr}{\mathrm{Tr}}

\newcommand{\proj}[1]{\ensuremath{|#1\rangle \langle #1|}}

\renewcommand{\rho}{\varrho}

\newcommand{\Hi}{\mathcal{H}}

\newcommand{\hi}{\Hi}

\newcommand{\one}{\mathds 1}

\newcommand{\secpar}{\ensuremath{n}\xspace}
\newcommand{\prob}{\ensuremath{\mathsf{Pr}}}
\renewcommand{\prob}{\Pr}

{\begin{framed}\begin{small}}
{\end{small}\end{framed}}

\newcommand{\SKQES}{\ensuremath{\textsf{SKQES}}\xspace}
\newcommand{\CSKE}{\ensuremath{\textsf{SKES}}\xspace}
\newcommand{\SKES}{\ensuremath{\textsf{SKES}}\xspace}

\newcommand{\expref}[2]{\texorpdfstring{\hyperref[#2]{#1~\ref{#2}}}{#1~\ref{#2}}}

\newcommand{\A}{\ensuremath{\mathcal{A}}\xspace}
\newcommand{\B}{\ensuremath{\mathcal{B}}\xspace}

\newcommand{\D}{\ensuremath{\mathcal{D}}\xspace}

\renewcommand{\H}{\ensuremath{\mathcal H}\xspace}
\newcommand{\K}{\ensuremath{\mathcal{K}}\xspace}
\newcommand{\M}{\ensuremath{\mathcal{M}}\xspace}

\newcommand{\states}{\ensuremath{\mathfrak D}\xspace}
\newcommand{\negl}{\ensuremath{\operatorname{negl}}\xspace}
\newcommand{\supp}{\textbf{supp\,}}
\newcommand{\egoketbra}[1]{\ketbra{#1}{#1}}			
\newcommand{\from}{\ensuremath{\leftarrow}}
\newcommand{\bit}{\{0,1\}}

\newcommand{\rand}{\raisebox{-1pt}{\ensuremath{\,\xleftarrow{\raisebox{-1pt}{$\scriptscriptstyle\$$}}\,}}}
\newcommand{\bin}{\left\{0,1\right\}}

\newcommand{\KeyGen}{\ensuremath{\mathsf{KeyGen}}\xspace}
\newcommand{\Enc}{\ensuremath{\mathsf{Enc}}\xspace}
\newcommand{\Dec}{\ensuremath{\mathsf{Dec}}\xspace}

\newcommand{\poly}{\operatorname{poly}}
\newcommand{\algo}{\mathcal}
\newcommand{\inrand}{\raisebox{-1pt}{\ensuremath{\,\xleftarrow{\raisebox{-1pt}{$\scriptscriptstyle\$$}}\,}}}

\newcommand{\PRF}{\ensuremath{\mathsf{PRF}}\xspace}
\newcommand{\QPRF}{\ensuremath{\mathsf{QPRF}}\xspace}
\newcommand{\pqPRF}{\ensuremath{\mathsf{pqPRF}}\xspace}

\newcommand{\des}{\ensuremath{\mathsf{2desTag}}\xspace}

\newcommand{\onePRF}{\ensuremath{\mathsf{1des^{PRF}}}\xspace}

\newcommand{\adver}{\A}
\newcommand{\badver}{\B}
\newcommand{\chall}{\algo C}

\newcommand{\QUFforge}{\ensuremath{\mathsf{QUF\mbox{-}Forge}}\xspace}
\newcommand{\QUFcheat}{\ensuremath{\mathsf{QUF\mbox{-}Cheat}}\xspace}
\newcommand{\UFforge}{\ensuremath{\mathsf{UF\mbox{-}Forge}}\xspace}
\newcommand{\UFcheat}{\ensuremath{\mathsf{UF\mbox{-}Cheat}}\xspace}
\newcommand{\QAEreal}{\ensuremath{\mathsf{QAE\mbox{-}Real}}\xspace}
\newcommand{\QAEideal}{\ensuremath{\mathsf{QAE\mbox{-}Ideal}}\xspace}
\newcommand{\QINDCCAAtest}{\ensuremath{\mathsf{QCCA2\mbox{-}Test}}\xspace}
\newcommand{\QINDCCAAfake}{\ensuremath{\mathsf{QCCA2\mbox{-}Fake}}\xspace}
\newcommand{\INDCCAAtest}{\ensuremath{\mathsf{CCA2\mbox{-}Test}}\xspace}
\newcommand{\INDCCAAfake}{\ensuremath{\mathsf{CCA2\mbox{-}Fake}}\xspace}

\newcommand{\IND}{\ensuremath{\mathsf{IND}}\xspace}
\newcommand{\INDCPA}{\ensuremath{\mathsf{IND\mbox{-}CPA}}\xspace}
\newcommand{\INDCCA}{\ensuremath{\mathsf{IND\mbox{-}CCA1}}\xspace}
\newcommand{\INDCCAA}{\ensuremath{\mathsf{IND\mbox{-}CCA2}}\xspace}

\newcommand{\QIND}{\ensuremath{\mathsf{QIND}}\xspace}
\newcommand{\QINDCPA}{\ensuremath{\mathsf{QIND\mbox{-}CPA}}\xspace}
\newcommand{\QINDCCA}{\ensuremath{\mathsf{QIND\mbox{-}CCA1}}\xspace}
\newcommand{\QINDCCAA}{\ensuremath{\mathsf{QIND\mbox{-}CCA2}}\xspace}

\newcommand{\EUFCMA}{\ensuremath{\mathsf{EUF\mbox{-}CMA}}\xspace}

\newcommand{\INTCTXT}{\ensuremath{\mathsf{INT\mbox{-}CTXT}}\xspace}
\newcommand{\QUF}{\ensuremath{\mathsf{QUF}}\xspace}
\newcommand{\UF}{\ensuremath{\mathsf{UF}}\xspace}
\newcommand{\CAE}{\ensuremath{\mathsf{AE}}\xspace}
\newcommand{\QAE}{\ensuremath{\mathsf{QAE}}\xspace}
\newcommand{\QCA}{\ensuremath{\mathsf{QCA}}\xspace}
\newcommand{\cQCA}{\ensuremath{\mathsf{cQCA}}\xspace}
\newcommand{\CAEreal}{\ensuremath{\mathsf{AE\mbox{-}Real}}\xspace}
\newcommand{\CAEideal}{\ensuremath{\mathsf{AE\mbox{-}Ideal}}\xspace}

\newcommand{\DNS}{\ensuremath{\mathsf{DNS}}\xspace}
\newcommand{\cDNS}{\ensuremath{\mathsf{cDNS}}\xspace}

\newcommand{\acc}{\ensuremath{\mathsf{acc}}\xspace}
\newcommand{\crej}{\ensuremath{\mathsf{rej}}\xspace}
\newcommand{\rej}{\ensuremath{\mathsf{reject}}\xspace}

\newcommand{\win}{\ensuremath{\mathsf{win}}\xspace}
\newcommand{\cheat}{\ensuremath{\mathsf{cheat}}\xspace}
\newcommand{\real}{\ensuremath{\mathsf{real}}\xspace}
\newcommand{\ideal}{\ensuremath{\mathsf{ideal}}\xspace}

\title{Unforgeable Quantum Encryption
}
\author{} \institute{}
\date{\today}
\author{Gorjan Alagic\inst{1,2} \and Tommaso Gagliardoni\inst{3} \and Christian Majenz\inst{4,5}}\institute{}
\institute{Joint Center for Quantum Information and Computer Science, University of Maryland, College Park, MD \and National Institute of Standards and Technology, Gaithersburg, MD\and IBM Research, Zurich, Switzerland\and Institute for Logic, Language and Computation, University of Amsterdam, Amsterdam, Netherlands\and Centrum for Wiskunde en Informatica, Amsterdam, Netherlands
	\\ \email{galagic@umd.edu}; \email{tog@zurich.ibm.com}; \email{c.majenz@uva.nl}}

\makeindex

\begin{document}

\maketitle

\begin{abstract}
We study the problem of encrypting and authenticating quantum data in the presence of adversaries making adaptive chosen plaintext and chosen ciphertext queries. Classically, security games use string copying and comparison to detect adversarial cheating in such scenarios. Quantumly, this approach would violate no-cloning. We develop new techniques to overcome this problem: we use entanglement to detect cheating, and rely on recent results for characterizing quantum encryption schemes. We give definitions for (i.) ciphertext unforgeability
, (ii.) indistinguishability under adaptive chosen-ciphertext attack, and (iii.) authenticated encryption. 
The restriction of each definition to the classical setting is at least as strong as the corresponding classical notion: (i) implies \INTCTXT, (ii) implies \INDCCAA, and (iii) implies \CAE. All of our new notions also imply \QINDCPA privacy. Combining one-time authentication and classical pseudorandomness, we construct symmetric-key quantum encryption schemes for each of these new security notions, and provide several separation examples. Along the way, we also give a new definition of one-time quantum authentication which, unlike all previous approaches, authenticates ciphertexts rather than plaintexts.

\end{abstract}

\section{Introduction}

Given the rapid development of quantum information processing, it is reasonable to conjecture that future communication networks will include at least some large-scale quantum computers and high-capacity quantum channels. What will secure communication look like on the resulting ``quantum Internet''? For instance, how will we transmit quantum messages securely over a completely insecure channel? One approach is via interactive and information-theoretically secure methods, e.g., combining entanglement distillation with teleportation. In this work, we will instead consider the non-interactive, highly efficient approach which dominates the current classical Internet. A natural goal here is to achieve, in the quantum setting, all the basic features that are enjoyed by classical encryption: (i.) a single small key suffices for transmitting an essentially unlimited amount of data, (ii.) these keys can be exchanged over public channels, and (iii.) the security guarantees are as strong as possible. Previous work has shown how to achieve both (i.) and (ii.), but only for secrecy against chosen-plaintext and non-adaptive chosen-ciphertext attacks~\cite{BJ15,ABF+16}. Authentication or adaptive chosen-ciphertext security for such schemes has, as yet, not been considered. In fact, at the time of writing, there is not even a definition for \emph{two-time quantum authentication}, much less for quantum analogues of \EUFCMA or \INDCCAA. The aim of this work is to address this problem.

The security definitions we seek do not yet exist due to a number of technical obstacles, all of which can be traced to quantum no-cloning and the destructiveness of quantum measurements. These obstacles make it difficult even just to formulate the basic security notion, much less to prove reductions or to construct secure schemes. In unforgeability, for example, no-cloning makes it impossible to record the adversary's queries and check whether the final output is a fresh forgery. In adaptive chosen-ciphertext security, no-cloning makes it impossible to record the challenge ciphertext and ensure that the adversary does not ``cheat'' by simply decrypting it (and thus win against any scheme). Moreover, due to the destructiveness of quantum measurement, it is unclear if one can \emph{both} perform cheat-detection \emph{and} answer non-cheating queries correctly. 

In this work, we overcome these obstacles, and present the first definitions of multiple-query unforgeability and adaptive chosen-ciphertext indistinguishability for quantum encryption schemes, thereby solving a longstanding open problem~\cite{ABF+16,BZ13,GHS16}. While our definitions are inherently quantum in nature, we are able to show that they are in fact natural analogues of well-known classical security definitions, such as \INTCTXT and \INDCCAA. The strongest security notion we define is called \emph{quantum authenticated encryption} (or \QAE) and corresponds to the strongest form of security normally studied in the classical setting. A secret-key scheme satisfying \QAE is unforgeable and indistinguishable even against adversaries that can make adaptive encryption and decryption queries. 

In an effort to explore this new landscape, we prove several theorems which relate our new notions to each other and to established quantum and classical security definitions. We also show how to satisfy each of our new security notions with explicit, efficient constructions. In particular, we show that combining a post-quantum secure pseudorandom function with a unitary $2$-design yields the strongest form of secret-key quantum encryption defined thus far, i.e., \QAE.

\subsubsection{Related Work.}

Computationally-secure quantum encryption has garnered significant interest in the past few years, beginning with basic security notions  like \QINDCPA and \QINDCCA~\cite{BJ15,ABF+16}, and then with more advanced concepts such as quantum fully-homomorphic encryption (QFHE) ~\cite{BJ15,DSS16}. For authentication, uncloneability, and non-malleability, the one-time setting has received considerable attention~(see, e.g., \cite{ABW09,DNS12,HLM16,GYZ17,AM17,BW16,Portmann17,Gottesman03}.)  We will make use of the authentication definition of~\cite{DNS12}, a characterization lemma of~\cite{AM17}, and a simulation adversary of~\cite{BW16}. 
For classical notions of unforgeability and chosen-ciphertext security, see e.g.~\cite{KL14}.

\subsection{Our approach}

\subsubsection{The problem.}\label{sec:problem}

We begin by outlining the technical difficulties in some further detail. 
Let us consider many-time authentication for symmetric-key encryption schemes first.
In the classical setting, secure many-time authentication is defined in terms of \emph{unforgeability}. A scheme is unforgeable if no adversary, even if granted the black-box power to authenticate with our secret key, can generate a fresh and properly authenticated message (i.e., a forgery). Translating this idea to the quantum setting presents immediate technical difficulties. First, no-cloning prevents us from recording the adversary's previous queries. Second, even if the first problem is surmounted, the nature of measurement might make it difficult to reliably identify whether the adversary's output is indeed fresh. For example, we might need many copies of the adversary's query, as well as many copies of their final output. 

A similar problem occurs for secrecy. The current state-of-the-art is the so-called \QINDCCA model. In this model, the transmitted state (the ``challenge'') remains secret even to adversaries with the black-box power to both encrypt and non-adaptively decrypt with our secret key. Our experience in the classical world tells us that this model is too weak, because real-world adversaries can sometimes gain \emph{adaptive} access to decryption (e.g., in WEP and early versions of SSL~\cite{Barak-notes-ch6}.) Classically, this is addressed using the so-called \INDCCAA model, where the adversary is allowed adaptive decryption queries \emph{but cannot use them on the challenge} (without this caveat, security becomes impossible). Here again, the quantum setting presents numerous technical difficulties: no-cloning prevents us from recording the challenge, and the nature of measurement makes it difficult to tell if the adversary is attempting to decrypt the challenge.

Recall that the strongest form of classical security, so-called ``authenticated encryption'' (or \CAE) is defined to be \INDCCAA together with unforgeability of ciphertexts~\cite{KL14}. Achieving a comparable quantum notion thus seems to require solving all of the above problems.

Using classical intuition, one might attempt a solution as follows: consider only pure-state plaintexts, and demand that the final forgery is orthogonal to the previous queries (or, in $\textsf{CCA2}$, that decryption queries are orthogonal to the challenge). This may seem promising at first, but a closer look reveals numerous issues; for example: (i.) quantum states are in general not pure, and may include side registers kept by the adversary, (ii.) this idea charges the adversary with adhering to very strict demands, contrary to good theory practice, (iii.) checking whether a particular adversary satisfies the demands cannot be done efficiently.

\subsubsection{A promising approach.}\label{sec:approach}

We now describe a more promising solution, beginning with unforgeability. We will express security in terms of the performance of adversaries $\algo A$ in two games: (1.) \textsf{F-Real}, where $\algo A$ gets oracle access to $\Enc_k$ and wins if he outputs {\em any} valid ciphertext, and (2.) \textsf{F-Cheat}, where we attempt to ascertain if $\algo A$ is cheating by feeding us an output of the oracle. How do we detect this kind of cheating? Recall that, even in the one-time setting, quantum authentication implies indistinguishability of ciphertexts. A consequence of this is that, whenever $\algo A$ performs an encryption query on a certain plaintext state, we are free to respond with an encryption \emph{of a different state} -- for example, half of a maximally-entangled state. This will be our approach: we prepare an entangled pair $\ket{\phi^+}_{MM'}$, apply $\Enc_k$ to register $M$, give the resulting ciphertext register to $\algo A$, and keep $M'$. When the game ends, we decrypt the output of $\algo A$ into a register $O$, and then perform the measurement $\{\Pi_{\phi^+}, \one - \Pi_{\phi^+}\}$ on $OM'$. We then declare that $\algo A$ is cheating if and only if the first outcome is recorded. 

This idea can also be applied to the multiple-query setting. There, we respond to the $j$th query with an encryption of register $M$ of $\ket{\phi^+}_{MM_j}$, and save $M_j$; at the end of the game, we perform the aforementioned measurement on $OM_j$ for all $j$ and declare that $\algo A$ cheated if any of them return the first outcome.

To define a quantum analogue of \INDCCAA, we can try a similar strategy. We again compare the performance of $\algo A$ in two games: (1.) \textsf{C-Real}, which is just like the classical \INDCCAA game, except with no restrictions on $\algo A$'s use of the $\Dec_k$ oracle, and (2.) \textsf{C-Cheat}, where we again attempt to detect cheating. In \textsf{C-Cheat}, when the adversary sends us the challenge plaintext, we discard it and respond with the ciphertext register of $(\Enc_k \otimes \one_{M'})\ket{\phi^+}_{MM'}$ instead, while keeping $M'$ to ourselves. Whenever $\algo A$ queries the decryption oracle, we first apply $\Dec_k$ and place the resulting plaintext in a register $O$. Then we apply the measurement $\{\Pi_{\phi^+}, \one - \Pi_{\phi^+}\}$ to $OM'$ to see if the adversary is cheating. If we get the first outcome, we declare that $\algo A$ cheated.

The above ideas do lead to reasonable security definitions, which (at least partly) fulfill our original goals. However, they suffer from a number of drawbacks. First, repeated measurement of the plaintext requires the use of a so-called ``gentle measurement lemma''~\cite{Winter99}, and thus can only apply to large plaintext spaces
 (e.g., $n^c$ qubits for $c > 0$). Second, they only offer \emph{plaintext authentication} and a kind-of \emph{plaintext CCA security}; modification of ciphertexts (that does not also modify the underlying plaintext) cannot be detected. Our classical experience tells us that this is insufficient, and that we should demand impossibility of any ciphertext manipulation whatsoever. Addressing these problems is where many of our new technical contributions (in addition to the above ideas) are needed. While our actual approach will be different, and more sophisticated techniques are required, we will still follow the spirit of the idea outlined above.

\subsection{Summary of Results}

Recall that, in the setting of quantum data, copying is impossible and authentication implies encryption~\cite{BCG+02}. In particular, there is no direct quantum analogue of a MAC. As a result, the central objects of study in our work will be symmetric-key quantum encryption schemes, or \SKQES for short, but our results on quantum CCA2 security 
carry over to the public-key setting as well.

\subsubsection{Quantum ciphertext authentication.}
All previous definitions of authentication for quantum data allow manipulation of the ciphertext (see \expref{Section}{sec:dns}), thus only authenticating the plaintext state. In our first main contribution, we solve this problem, laying the necessary groundwork for our remaining results.
\begin{itemize}
\item We give a new definition: information-theoretic \emph{quantum one-time ciphertext authentication} (\QCA), inspired by ideas of~\cite{AM17,BW16}.
\item We prove that \QCA is a strengthening of ``\DNS''-authentication~\cite{DNS12}.
\begin{theorem}[informal]
	If a \SKQES authenticates ciphertexts \emph{(\QCA)}, then it also authenticates plaintexts \emph{(\DNS)}; in particular, it satisfies secrecy \emph{(\QIND)}.
\end{theorem}
\item We define computational-security (one-time) analogues: \cQCA and \cDNS.
\end{itemize}

\subsubsection{Quantum unforgeability.}
In this setting, the adversary is granted access to an encryption oracle, and must generate a valid ``fresh'' ciphertext. 
\begin{itemize}
\item We give a new definition: \emph{quantum unforgeability} (\QUF), combining ideas of~\expref{Section}{sec:approach} and~\cite{AM17}. We also define a bounded-query analogue ($t$-\QUF).
\item We show that \UF, the classical analogue of \QUF, is remarkably strong.
\begin{theorem}[informal]
For classical schemes, $\UF \iff \CAE$.
\end{theorem}
\end{itemize}

\subsubsection{Quantum chosen-ciphertext security.}
We address the longstanding problem of defining quantum security under adaptive chosen-ciphertext attack~\cite{ABF+16,BZ13,GHS16}; the state of the art was previously the non-adaptive \QINDCCA~\cite{ABF+16}.
\begin{itemize}
\item We give a new definition: \emph{quantum indistinguishability under adaptive chosen-ciphertext attack}\\ (\QINDCCAA), using all of the aforementioned ideas.
\item We relate \QINDCCAA to existing security notions.
\begin{theorem}[informal]
\begin{enumerate}
\item For quantum schemes, $\QINDCCAA \implies \QINDCCA$.
\item The classical analogue of $\QINDCCAA$ is equivalent to classical \INDCCAA.
\end{enumerate}
\end{theorem}
\end{itemize}

\subsubsection{Quantum authenticated encryption.}
In our main contribution, we define a natural quantum analogue of the classical concept of \emph{authenticated encryption} (\CAE). All previous quantum security notions lacked both unforgeability and adaptive chosen-ciphertext security.
\begin{itemize}
\item We give a new definition: \emph{quantum authenticated encryption} (\QAE), combining the ideas of \expref{Section}{sec:approach}, the notion of \QCA, and a real/ideal approach~\cite{Shrimpton04}.
\item We give evidence that \QAE is indeed the correct quantum analogue of \CAE.
\begin{theorem}[informal]\label{thm:implications-informal} 
~
\begin{enumerate}
\item Unforgeability and secure authentication: $\QAE \implies \QUF \wedge \cQCA$. 
\item Chosen-ciphertext security: $\QAE \implies \QINDCCAA$.
\item The classical analogue of \QAE is equivalent to classical \CAE.
\end{enumerate}
\end{theorem}
\end{itemize}
\begin{figure}[h]
\centering
\includegraphics[width=0.80\textwidth]{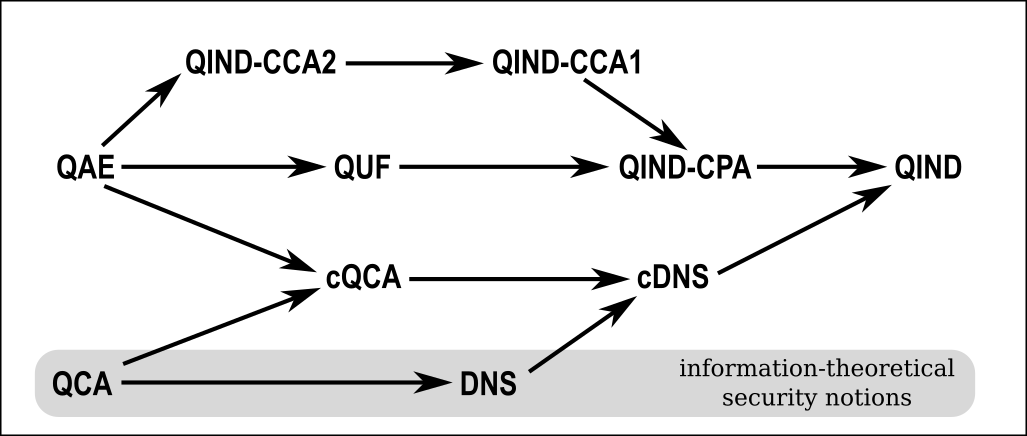}
\caption{Implications between quantum security notions}
\label{fig:impl}
\end{figure}
The new notions and connections we develop are summarized in \expref{Figure}{fig:impl}.

\subsubsection{Constructions and separations.}

Our new constructions combine a \SKQES $\Pi$ with a classical keyed function family $f$ to build a new \SKQES $\Pi^f$, as follows. In $\Pi^f$, key generation outputs a key for $f$; to encrypt a state $\rho$, we generate a random $r$ and output $(r, \Enc^\Pi_{f_k(r)}(\rho))$.  For example, if $\Pi$ is the quantum one-time pad and $f$ is a \pqPRF (i.e, a post-quantum-secure pseudo-random function), then $\Pi^f$ is the \INDCCA-secure scheme from~\cite{ABF+16}. We will also need the standard one-time authentication scheme $\des$, defined by $\Enc_k: \rho \mapsto C_k (\rho \otimes \egoketbra{0^n}) C_k^\dagger$ where $C$ is an (exact or approximate) unitary two-design. 

\begin{theorem}[informal]
Let $\Pi$ be a $\des$ scheme, let $f$ be a \pqPRF, and let $g$ be a $t$-wise independent classical function family. Then
\begin{enumerate}
\item $\Pi$ is one-time ciphertext authenticating (\QCA).
\item $\Pi^g$ is $t$-time quantum unforgeable ($t$-\QUF).
\item $\Pi^f$ satisfies quantum authenticated encryption (\QAE); in particular, it  is quantum unforgeable (\QUF) and chosen-ciphertext secure (\QINDCCAA).
\end{enumerate}
\end{theorem}

\begin{theorem}[informal]
~
\begin{enumerate}
\item There exists an \SKQES which is \QINDCCA but not \QINDCCAA.
\item There exists an \SKQES which is \QINDCCAA but not \QAE.
\end{enumerate}
\end{theorem}

\subsubsection{Our choice of primitives.}
The reader may wonder why our constructions do not need ``quantum-oracle-secure'' primitives (e.g., $\QPRF$s for unforgeability and $2t$-wise independence for $t$-time security, as in the quantum-secure classical setting of~\cite{BZ13a}.) In our work, the classical portion of the ciphertext is generated by honest parties during encryption, and measured during decryption. As a result, oracle access to $\Enc_k$ and $\Dec_k$ (as CPTP maps) never grants quantum oracle access to the underlying classical primitive. Of course, one could grant the adversary more powerful oracles that do grant this kind of access, and then quantum-oracle-secure primitives (such as $\QPRF$s) would indeed be required.

\subsubsection{A remark on applicability.}
While all of our definitions apply to arbitrary quantum encryption schemes, security reductions sometimes require the following additional condition. As discussed in~\expref{Section}{sec:characterization}, all quantum encryption algorithms can be characterized as (1.) drawing a random pure state from a probability distribution, (2.) attaching it to the plaintext, and (3.) applying a unitary operator. For the implication \QAE$\Rightarrow\cQCA$ of \expref{Theorem}{thm:implications-informal} to hold, it is required that (1), (2) and (3) are efficiently implementable. This condition holds for all schemes known to us. However, it is \emph{in principle} possible that there are schemes for which $\Enc_k$ is efficiently implementable, but the particular implementation ``(1), then (2), then (3)'' is not. We leave this as an open problem.


\section{Preliminaries}

\subsubsection{Basic Notation and Conventions.}

In the rest of this work, we use ``classical'' to denote ``non-quantum'', ``iff'' for ``if and only if'', and \secpar to denote the security parameter. A function $\varepsilon(\secpar)$ is negligible (denoted $\varepsilon(\secpar) \leq \negl(\secpar)$) if it is asymptotically smaller than $1/p(\secpar)$ for every polynomial function $p$. The notation $x \rand X$ means that $x$ is a sample from the uniform distribution over the set $X$. By ``PPT'' we mean a polynomial-time uniform family of probabilistic circuits, and by ``QPT'' we mean a polynomial-time uniform family of quantum circuits. We will frequently give such algorithms names like ``adversary'' or ``challenger,'' but this is only to help remember the role of the algorithm.

For notation and conventions regarding quantum information, we refer the reader to~\cite{NC11}. We recall a few basics here. We denote by $\hi_M$ a complex Hilbert space with label $M$ and finite dimension $\dim M$. We use the standard bra-ket notation to work with pure states $\ket{\varphi} \in \hi_M$. The class of positive, Hermitian, trace-one linear operators on $\hi_M$ is denoted by $\states(\hi_M)$. A {\em quantum register} is a physical system whose set of valid states is $\states(\hi_M)$; in this case we label by $M$ the register itself. We reserve the notation $\tau_M$ for the maximally mixed state (i.e., uniform classical distribution) $\one / \dim M$ on $M$.

In a typical cryptographic scenario, a ``quantum register $M$'' is in fact an infinite family of registers $\{M_n\}_{n \in \NN}$ consisting of $p(n)$ qubits, where $p$ is some fixed polynomial. This family is parameterized by $n$, which is typically also the security parameter. We will consider completely positive (CP), trace-preserving (TP) maps (i.e., quantum channels) when describing quantum algorithms. To indicate that $\Phi$ is a channel from register $A$ to $B$, we will write $\Phi_{A \to B}$. When it helps to clarify notation, we will use $\circ$ to denote composition of operators. We will also often drop tensor products with the identity, e.g., given a map $\Psi_{BC \to D}$, we will write $\Psi \circ \Phi$ to denote the map $\Psi \circ (\Phi \otimes \one_C)$ from $AC$ to $D$. 

\label{par:support}
The support of a quantum state $\rho$ is its cokernel (as a linear operator). Equivalently, this is the span of the pure states making up any decomposition of $\rho$ as a convex combination of pure states. We will denote the orthogonal projection operator onto this subspace by $P_\rho$. The two-outcome projective measurement (to test if a state has the same or different support as $\rho$) is then $\{P_\rho, \one -P_\rho\}$.

Next, we single out some unitary operators that will appear frequently. First, the group of $n$-qubit operators generated by Paulis $I, X, Y, Z$ (applied to individual qubits) is a well-known \emph{unitary one-design}. The Clifford group on $n$ qubits is defined to be the normalizer of the Pauli group inside the unitary group. It can also be seen as the group generated by the gate set $(H, P, CNOT)$~\cite{Gottesman98}; it is also a \emph{unitary two-design}~\cite{VLT02}. 

A \emph{unitary $t$-design} (for a fixed $t$) is an infinite collection $\mathcal U = \{\mathcal U^{(n)} : n \in \N\}$, where $\mathcal U^{(n)}$ forms an $n$-qubit unitary t-design in the standard sense, i.e.,
\begin{equation}
\frac{1}{|\mathcal U^{(n)}|}\sum_{U\in {\mathcal U}^{(n)}}U^{\otimes t}X\left(U^\dagger\right)^{\otimes t}
= \int U^{\otimes t}X\left(U^\dagger\right)^{\otimes t} dU\,.
\end{equation}
In the above, the integral is taken over the $n$-qubit unitary group according to the Haar measure. We assume that there is an explicit polynomial function $m(n)$ and a deterministic polynomial-time algorithm which, given $1^n$ and $k \inrand \bit^{m(n)}$, produces a circuit for a unitary operator $U_{k, n}$ which is distributed uniformly at random in $\mathcal U^{(n)}$. We will not refer to this algorithm explicitly and will simply write $\{U_{k, n} : k \in \bit^{m(n)}\}$ for the resulting distribution on unitary operators; we will also frequently suppress one index and write $U_k$ when $n$ is clear from context. We refer to the polynomial $m$ as the \emph{key length} of the $t$-design. Standard examples are: (i.) the Pauli one-design (where we apply $X^aZ^b$ to each qubit for random $a, b \in \bit$) is a unitary one-design on $n$ qubits with key length $2n$; (ii.) the Clifford group (where we apply a uniformly random element of the $n$-qubit Clifford group, efficiently generated via the Gottesman-Knill theorem~\cite{AG04}) is a unitary 3-design, and therefore in particular a unitary 2-design, on $n$ qubits with key length $O(n^2)$; (iii.) random $\poly(t, n)$-size quantum circuits, randomly generated from a universal gate set, are approximate $t$-designs on $n$ qubits~\cite{BHH16}.

In this work, we will only require one-designs and two-designs, and we will assume for simplicity that the designs are exact. While approximate designs would also suffice, some additional (but straightforward) analysis would be required.

\subsubsection{Quantum Encryption.}\label{sec:basics}

We will follow the conventions set in~\cite{ABF+16}; the exception is that decryption can reject by outputting a special symbol $\bot$. 

\begin{definition}\label{def:SKE}
A {\em symmetric-key quantum encryption scheme (or $\SKQES$)} is a triple of QPT algorithms:
\begin{enumerate}
\item (key generation)\footnote{A more general definition uses arbitrary key generation algorithms. We assume a uniform key in this paper for technical and notational convenience.} $\KeyGen:$ on input $1^n$, outputs $k \rand \K$
\item (encryption) $\Enc: \K \times \states(\H_M) \rightarrow \states(\H_C)$
\item (decryption) $\Dec: \K \times \states(\H_C) \rightarrow \states(\H_M \oplus \ketbra{\bot}{\bot})$
\end{enumerate}
such that $\| \Dec_k \circ \Enc_k - \one_M \oplus 0_\bot \|_\diamond \leq \negl(n)$
for all $k \in \emph{\supp} \KeyGen(1^n)$. 
\end{definition}

It is implicit that the key space $\K$ is classical and of size $\poly(n)$; likewise, the registers $C$ and $M$ are quantum registers of at most $\poly(n)$ qubits. We will only consider \SKQES of \emph{fixed-length}, meaning that the number of qubits in $M$ is a fixed function of the security parameter $n$. We assume that honest parties will apply the measurement $\{\Pi_\bot , \one - \Pi_\bot\}$ (where $\Pi_\bot = \egoketbra{\bot}$) immediately after decryption. This allows us to write, e.g., $\Dec_k(\rho) \neq \bot$ to mean that decryption (followed by this measurement) successfully produced a valid plaintext.

We will often combine quantum schemes with classical (keyed) function families. A keyed function family consists of functions $f:\bit^{p(n)} \times \bit^{q(n)} \to \bit^{s(n)}$ where $p, q, s$ are polynomials in $n$. In typical usage, we sample a key $k \rand \bit^{p(n)}$ and then consider the restricted function $f_k : \bit^{q(n)} \to \bit^{s(n)}$ defined by $f_k(x) = f(k, x)$. All keyed function families are assumed to be computable by a deterministic polynomial-time uniform classical algorithm. 

\begin{definition}\label{def:SKE-generic}
Let $\Pi = (\KeyGen^\Pi, \Enc^\Pi, \Dec^\Pi)$ be a \SKQES, and $f:\bit^{p(n)} \times \bit^{q(n)} \to \bit^{s(n)}$ a classical keyed function family. Define a new \SKQES $\Pi^f = (\KeyGen, \Enc, \Dec)$ as follows:
\begin{enumerate}
	\item $\KeyGen:$ on input $1^n$, outputs $k \rand \bit^{p(n)}$;
	\item $\Enc_k:$ on input $\rho$, outputs  $\ket{r}\bra{r} \otimes \Enc^\Pi_{f_k(r)}(\rho)$, where $r \rand \bit^{q(n)}$;
	\item $\Dec_k: \egoketbra{s} \otimes \sigma \mapsto \Dec^\Pi_{f_k(s)}(\sigma)$.
\end{enumerate}
\end{definition}
We extend $\Dec_k$ to arbitrary inputs by postulating that it begins by measuring the first register in the computational basis. Note that $\Pi^f$ has plaintext length $t(s(n))$ where $t(.)$ is the plaintext length of $\Pi$ as a function of $\Pi$'s key length. 
This construction can be extended to schemes $\Pi$ with a non-uniform key by using the output of the keyed function family as a random tape for $\KeyGen^\Pi$.

\subsubsection{Quantum secrecy.}

The literature contains a number of information-theoretic definitions of quantum secrecy~(see, e.g., \cite{AMT+00,ABW09,BJ15,ABF+16}). It is well-known that a unitary one-design (e.g., the Pauli group) is an information-theoretically secret scheme. In this work, however, we focus on the computational setting~\cite{BJ15,ABF+16}.

\begin{definition}[\QIND]\label{def:IND}
A $\SKQES$ $\Pi = (\KeyGen, \Enc, \Dec)$ has indistinguishable encryptions (or is \QIND) if for every QPT adversary $\A=(\M,\D)$ we have:
$$
\Bigl| \Pr \bigl[ \D \big\{ (\Enc_k \otimes \one_E) \rho_{ME} \big\}  = 1  \bigr]
- \Pr \bigl[ \D \big\{ (\Enc_k \otimes \one_E) (\egoketbra{0}_M \otimes \rho_E) \big\} = 1 \bigr] \Bigr| \leq \negl(n),
$$
where $\rho_{ME} \from \M(1^n)$, $\rho_E = \tr_M(\rho_{ME})$, and the probabilities are taken over $k \leftarrow \KeyGen(1^n)$ and the coins and measurements of \Enc, $\M$, $\D$. We also define:
\begin{itemize}
\item \QINDCPA: In addition to the above, $\M$ and $\D$ have oracle access to $\Enc_k$.
\item \QINDCCA: In addition to \QINDCPA, $\M$ has oracle access to $\Dec_k$.
\end{itemize}
\end{definition}

Recall that a \pqPRF (post-quantum pseudorandom function) is a classical, deterministic, efficiently computable keyed function family $\{f_k\}_k$ which appears random to QPT algorithms with classical oracle access to $f_k$ for uniformly random $k$. The strongest notion ($\QINDCCA$) is satisfied by $\Pi^f$ where $\Pi$ is a one-design  and $f$ is a \pqPRF~\cite{ABF+16}. We let \onePRF denote such schemes.

\subsubsection{One-time authentication.}\label{sec:dns}

We recall quantum authentication as defined by Dupuis et al.~\cite{DNS12}, and adapt it to our conventions. Given an attack map $\Lambda_{CB \to C\tilde B}$ on a scheme $\Pi = (\KeyGen, \Enc, \Dec)$ (where the adversary holds $B$ and $\tilde B$), we define the ``averaged effective plaintext map'' (or just ``effective map'') as follows.
$$
\Lambda^\Pi_{MB\to M\tilde B} := \mathbb{E}_{\,k \from \KeyGen(1^n)} \left[ \Dec_k \circ \Lambda \circ \Enc_k\right]\,.
$$
We then require that, conditioned on acceptance, this map is the identity on $M$.
\begin{definition}[\cite{DNS12}]\label{def:auth}
A $\SKQES$ $\Pi = (\KeyGen, \Enc, \Dec)$ is \DNS-authenticating if, for any CP-map $\Lambda_{CB\to C\tilde B}$, there exist CP-maps $\Lambda^{\acc}_{B\to \tilde B}$ and $\Lambda^{\crej}_{B\to \tilde B}$ that sum to a TP map, such that
	\begin{equation}
	\left\|\Lambda^\Pi_{MB\to M\tilde B} - 
	\left(\id_M\otimes\Lambda^{\acc}_{B\to \tilde B} + \proj\bot_M\otimes \Lambda^{\crej}_{B\to \tilde B}\right)\right\|_\diamond\le\negl(n)\,.
	\end{equation}
\end{definition}

An important observation is that this definition only provides for authentication of the plaintext state. To see that this cannot be ``ciphertext authentication,'' simply take a scheme which is $\DNS$ and change it so that (i.) an extra bit is added to the ciphertext during encryption, and (ii.) that same bit is ignored during decryption. The resulting scheme still satisfies $\DNS$, but the adversary can clearly forge ciphertexts by flipping the extra bit. A perhaps more compelling example just adds encoding (in some QEC code) after encryption, and decoding prior to decryption. The adversary is then free to modify ciphertexts with correctable errors without violating \DNS. We remark that, in this respect, the recent strengthening of \DNS due to Garg et al.~\cite{GYZ17} is no different: a scheme secure according to this stronger notion of authentication can be modified in the same way without losing security.

Next, we recall a standard one-time authentication scheme. We encrypt by appending $n$ ``tag'' qubits in the fixed state $\ket{0}$ and then applying a random element of a 2-design. 
Decryption first undoes the 2-design, then outputs the plaintext iff all tag qubits measure to $0$; otherwise it outputs $\bot$.

\begin{scheme}\label{scheme:2-design}
The scheme family \des is defined as follows. Select a unitary 2-design $\mathcal U$ with key length $m(\cdot)$, and define algorithms:
\begin{enumerate}
\item $\KeyGen$: on input $1^n$, output $k \rand \bit^{m(2n)}$;
\item $\Enc_k$: on input $\rho_M$, output $U_k(\rho_M \otimes \egoketbra{0^n}_T)U^\dagger_k$
\item $\Dec_k$: on input $\sigma_{MT}$, output
$$
\bra{0^n}_TU_k^\dagger \sigma_{MT} U_k\ket{0^n}_T
+ \tr \bigl[(\mathds 1-\egoketbra{0^n}_T)U_k^\dagger \sigma_{MT} U_k\bigr]\egoketbra{\bot}_M\,.
$$ 
\end{enumerate}
\end{scheme}

We chose $\des$ to have plaintext and tag length $n$. It is well-known that, for plaintexts of at most polynomial length and tags of length at least $n^c$, these schemes are \DNS-authenticating~\cite{ABE10,DNS12}. 
\section{One-Time Ciphertext Authentication}

One-time quantum authentication has been extensively studied~\cite{BCG+02,DNS10,DNS12,BW16,GYZ17,AM17}. As we observed above, all of these works concern \emph{plaintext authentication}, which ensures that manipulated ciphertexts decrypt to either the original plaintext or the reject symbol. Classical MACs, on the other hand, provide \emph{ciphertext authentication}, which ensures that any ciphertext manipulation whatsoever will result in rejection. This distinction is important; for instance, in classical \INDCCAA, the adversary can defeat plaintext-authenticating schemes by invoking the decryption oracle on a modified challenge ciphertext.

In this section we show how to define and construct ciphertext authentication in the quantum setting. These ideas will be crucial to defining more advanced notions (such as ciphertext unforgeability and adaptive chosen-ciphertext security) later in the paper. We start with the information-theoretical security setting, and then we discuss how to apply these notions to the computational setting.

\subsubsection{A characterization of encryption schemes.}\label{sec:characterization}

We recall a lemma from~\cite{AM17} stating that all \SKQES encrypt by (i.) attaching some (possibly key-dependent) auxiliary state, and (ii.) applying a unitary\footnote{If the dimension of the plaintext space does not divide the dimension of the ciphertext space, then we may need an isometry. In our case, all spaces are made up of qubits.} operator. Decryption undoes the unitary, and then checks if the support of the state in the auxiliary register has changed. We emphasize that this characterization follows from correctness only, and thus applies to all schemes. 
\begin{lemma}[Lemma B.9 in~\cite{AM17}, restated]\label{lem:SKQES-char}
	Let $\Pi=(\KeyGen, \Enc, \Dec)$ be a $\SKQES$. Then $\Enc$ and $\Dec$ have the following form:
	\begin{align*}
	&\Enc_k(X_M) = V_k\left(X_M \otimes (\sigma_k)_T\right)V_k^\dagger\nonumber\\
	&\Dec_k(Y_C) = \tr_{T}\left[P^{\sigma_k}_T \left(V_k^\dagger Y_CV_k\right) P^{\sigma_k}_T \right] +\hat D_k\left[\bar P^{\sigma_k}_T \left(V_k^\dagger Y_C V_k\right) \bar P^{\sigma_k}_T\right].
	\end{align*}
	Here, $\sigma_k$ is a state on register $T$, $P^{\sigma_k}_T$ and $\bar P^{\sigma_k}_T$ are the orthogonal projectors onto the support of $\sigma^{(k)}$ (see \expref{Section}{par:support}) and its complement (respectively), $V_k$ is a unitary operator, and $\hat D_k$ is a channel.
\end{lemma}
In practice, $\hat D_k$ (i.e., the map that is applied to any ciphertext outside of the range of $\Enc_k$) will just discard the state and replace it with $\bot$. Let us explain how the schemes we have seen so far fit into this characterization. For $\des$, $\sigma_k$ is simply the (key-independent) pure state $\proj{0^n}_T$, $V_k$ is the unitary operator of the two-design corresponding to key $k$, $P^{\sigma_k} = \proj{0^n}$, and $\hat D_k$ replaces the state with $\bot$. For $\onePRF$, $\sigma_k$ is the maximally mixed state $\tau$ (i.e., the classical randomness $r$ from \expref{Definition}{def:SKE-generic}), and $V_k$ is the controlled-unitary which applies a quantum one-time pad on the first register, controlled on the contents of the second register (using the $\pqPRF$ $f$), i.e., $\ket{x}\ket{r} \mapsto P_{f_k(r)}\ket{x} \ket {r}$. Decryption undoes the controlled unitary and never rejects, i.e., $P^{\sigma_k} = \one$. This corresponds to the fact that $\tau$ has full support.

By considering the spectral decomposition of the state $\sigma_k$ from \expref{Lemma}{lem:SKQES-char}, it is straightforward to show that encryption can always be implemented using unitary operators and only classical randomness. We state this fact as follows.
\begin{corollary}\label{cor:randomized-isometry}
Let $\Pi=(\KeyGen, \Enc, \Dec)$ be a $\SKQES$. Then for every $k$, there exists a probability distribution $p_k:\{0,1\}^t\to[0,1]$ and a family of quantum states $\ket{\psi^{(k,r)}}_T$ such that $\Enc_k$ is equivalent to the following algorithm:
\begin{enumerate}
\item sample $r \in \bit^t$ according to $p_k$;
\item apply the following map: 
$\Enc_{k;r}(X_M) =V_k \left(X_M \otimes \proj{\psi^{(k,r)}}_T \right) V_k^{\dagger}$.
\end{enumerate}
Here $V_k$ and $T$ are defined as in \expref{Lemma}{lem:SKQES-char}, and $t$ is the number of qubits in $T$.
\end{corollary}
For example, in the case of \des, the distribution is a point distribution and $\ket{\psi^{(k, r)}} = \ket{0^t}$. In $\onePRF$, the distribution is uniform and $\ket{\psi^{(k, r)}} = \ket{r}$.

It is important to remark here that, even if $\Enc_k$ is a polynomial-time algorithm, the functionally-equivalent algorithm provided by \expref{Corollary}{cor:randomized-isometry} may not be. We thus define the following.

\begin{condition}\label{con:efficient}
{\em Let $\Pi$ be a \SKQES, and let $p_k$, $\ket{\psi^{(k, r)}}$ and $V_k$ be as given in \expref{Corollary}{cor:randomized-isometry}. We say that {\em $\Pi$ satisfies \expref{Condition}{con:efficient}} if there exist efficient quantum algorithms for (i.) sampling from $p_k$, (ii.) preparing $\ket{\psi^{(k, r)}}$, and (iii.) implementing $V_k$, and this holds for all but a negligible fraction of $k$ and $r$.} 
\end{condition}

We are not aware of any examples of \SKQES that violate \expref{Condition}{con:efficient}.  In fact, in all schemes we will consider (including all schemes constructed via \expref{Definition}{def:SKE-generic}), the distribution $p_k$ and the states $\ket{\psi^{(k, r)}}$ are trivial to prepare, and the unitaries $V_k$ are implementable by poly-size quantum circuits. In any case, when \expref{Condition}{con:efficient} is required for a particular result, we will state this explicitly.

\subsubsection{Defining ciphertext authentication.}\label{sec:define-qCA}

We begin by outlining our approach. Fix an encryption scheme $\Pi$ with plaintext register $M$ and ciphertext register $C$. Let $\Lambda_{CB\to C\tilde B}$ be an attack map. Intuitively, we would like to decide whether to accept or reject conditioned on whether $\Lambda$ has changed the ciphertext. A possible approach would be to use the simulator from Theorem 5.1 in~\cite{BW16}: in the case of acceptance, this simulator\footnote{In~\cite{BW16}, this simulator was used to prove $\DNS$ security of the \des scheme. Here, we consider whether that simulator can be used to \emph{define} secure authentication. } ensures that $\Lambda$ is equivalent to $\one_C \otimes \Phi$ for some side-information map $\Phi_{B \to \tilde B}$. While this approach is on the right track, it is unnecessarily strong as a definition of security: it prevents the adversary from even looking at (or copying) classical parts of the ciphertext! This would place strange requirements on encryption. It would disallow constant classical messages (e.g., ``begin PGP message'') accompanying ciphertexts. It would also disallow a large class of natural schemes, including all schemes $\Pi^f$ from \expref{Section}{sec:basics}. This class has many schemes that (intuitively speaking) should be adequate for authenticating poly-many quantum ciphertexts, such as the case where $\Pi$ applies a random unitary and $f$ is a random function. 

The key to finding the middle ground lies in \expref{Corollary}{cor:randomized-isometry}: any scheme can be decomposed in a way that enables us to check separately whether the identity has been applied to the quantum part, and whether the classical register has changed. In effect, this will amount to an additional constraint over \DNS-authentication\footnote{One might also start from the authentication definitions of \cite{GYZ17,Portmann17} rather than \DNS. However, this is not necessary: these definitions' advantage over \DNS is in key recycling; our setting is non-interactive and has no back-channel for key recycling.} (\expref{Definition}{def:auth}), demanding extra structure from the simulator. 

Recall that an attack $\Lambda_{CB \to C\tilde B}$ on the scheme $\Pi$ defines the averaged effective plaintext map $\Lambda^\Pi_{MB\to M\tilde B} = \mathbb{E}_k[ \Dec_k \circ \Lambda \circ \Enc_k]$. We define ciphertext authentication as follows, using notation from \expref{Lemma}{lem:SKQES-char} and \expref{Corollary}{cor:randomized-isometry}.
\begin{definition}\label{def:qCauth}
	A \SKQES $\Pi=(\KeyGen, \Enc, \Dec)$ is \emph{ciphertext authenticating}, or \QCA, if for all CP-maps $\Lambda_{CB\to C\tilde B}$, there exists a CP-map $\Lambda^{\crej}_{B\to \tilde B}$ such that:
	\begin{equation}\label{eq:normal-constraint}
	\left\|\Lambda^\Pi_{MB\to M\tilde B} - 
	\left(\id_M\otimes\Lambda^{\acc}_{B\to \tilde B} + \proj\bot_M\otimes \Lambda^{\crej}_{B\to \tilde B}\right)\right\|_\diamond\le\negl(n),
	\end{equation}
	and $\Lambda^{\acc}_{B\to \tilde B}+\Lambda^{\crej}_{B\to \tilde B}$ is TP. Here
	$\Lambda^{\acc}_{B\to \tilde B}$ is given by:
	\begin{equation}\label{eq:extra-constraint}
	\Lambda^{\acc}_{B\to \tilde B} (Z_B) =\mathbb{E}_{k,r}\left[\bra{\Phi_{k, r}} V_k^\dagger\Lambda \left(\Enc_{k;r}\left(\phi^+_{MM'} \otimes Z_B \right)\right)V_k\ket{\Phi_{k, r}}\right]
	\end{equation}
	where 
	$\ket{\Phi_{k, r}} = \ket{\phi^+}_{MM'}\otimes\ket{\psi^{(k,r)}}_T.$
\end{definition}
Condition \eqref{eq:normal-constraint} is simply \DNS. It ensures that, in the accept case, the adversary performs the identity on the plaintext. Condition \eqref{eq:extra-constraint} demands that the rest of the action (i.e., on the side-information) is well-simulated by the following:
\begin{enumerate}
\item prepare a maximally entangled state $\phi^+_{MM'}$ and attach it to the input $B$;
\item run encryption, saving the classical randomness $r$ used (meaning that the tag register $T$ was prepared in the state $\ket{\psi^{(k, r)}}$);
\item apply decryption while conditioning on (i.) the plaintext still being maximally entangled with $M'$, and (ii.) register $T$ still containing $\ket{\psi^{(k, r)}}$;
\item output the contents of $\tilde B$.
\end{enumerate}
Note that this definition only adds further constraints to \DNS. Recalling that \DNS implies \QIND~\cite{BCG+02,GYZ17}, we thus have the following.
\begin{theorem}\label{thm:qCauth2DNS-IND}
If a \SKQES is \QCA, then it is also \DNS; in particular, it is \QIND.
\end{theorem}

It is not difficult to see that the security proof in Theorem 5.1 of~\cite{BW16} (for establishing \DNS of the Clifford scheme) actually applies to arbitrary 2-designs, and in fact proves \QCA and not only \DNS. We thus have that the scheme \des fulfills ciphertext authentication. For details on the separation between \QCA and \DNS, see the appendix of the full version of this paper~\cite{AGM17eprint}.

\subsubsection{Computational-security variant.}

We now briefly record a computational-security variant of one-time ciphertext authentication, which simply requires that all elements in Definition~\ref{def:qCauth} are efficient.

\begin{definition}\label{def:cqCauth}
	A \SKQES $\Pi=(\KeyGen, \Enc, \Dec)$ is {\em computationally ciphertext authenticating (\cQCA)} if, for any efficiently implementable attack map $\Lambda_{CB\to C\tilde B}$, the effective attack $\tilde\Lambda_{MB\to M\tilde B}$ is computationally indistinguishable from the simulator:
	\begin{equation}\label{eq:simulator1}
	\Lambda^{\mathrm{sim}}_{MB\to M\tilde B}=\id_M\otimes\Lambda^{\acc}_{B\to \tilde B}+\proj\bot_M\otimes \Lambda^{\rej}_{B\to \tilde B}.
	\end{equation}
	Here the simulator is given by:
	\begin{align}\label{eq:simulator2}
	\Lambda^{\acc}_{B\to \tilde B}&=\mathbb{E}_{k,r}\left[\bra{\Phi_{k, r}} V_k^\dagger\Lambda \left(\Enc_{k;r}\left(\phi^+_{MM'}\otimes(\cdot)_B\right)\right)V_k\ket{\Phi_{k, r}}\right]\text{ and}\nonumber\\
	\Lambda^{\rej}_{B\to \tilde B}&=\mathbb{E}_{k,r}\left[\tr\left(\one-\proj{\Phi_{k, r}}\right) V_k^\dagger\Lambda \left(\Enc_{k;r}\left(\phi^+_{MM'}\otimes(\cdot)_B\right)\right)V_k\right] \ ,
	\end{align}
	where:
	$\ket{\Phi_{k, r}} = \ket{\phi^+}_{MM'}\otimes\ket{\psi^{(k,r)}}_T.$
\end{definition}
Because we fix the form of the simulator in the reject case, the simulator is efficiently implementable just as in~\cite{BW16} for schemes that satisfy \expref{Condition}{con:efficient}.
It is straightforward to define a computational variant of \DNS~\cite{BW16}, which we denote by \cDNS. Given that Theorem~\ref{thm:qCauth2DNS-IND} only talks about computationally bounded quantum adversaries, it also applies to \cDNS. In particular we have the following.

\begin{proposition}\label{prop:cQCAcDNS-QIND}
If a \SKQES is \cQCA, then it is also \cDNS; in particular, it satisfies \QIND.
\end{proposition}

\section{Quantum Unforgeability}

Translating the standard classical intuition of ciphertext unforgeability to the quantum setting appears nontrivial. As we develop our approach, it will be useful to keep in mind a ``prototype'' scheme that should (intuitively) satisfy quantum unforgeability against a polynomial-time adversary making an arbitrary number of queries. This is the scheme $\des^\PRF$, which encrypts via:
$$
\Enc_k ( \rho) = U_{f_k(r)} \left( \rho \otimes \egoketbra{0^n} \right) U_{f_k(r)}^\dagger \otimes \egoketbra{r}
$$
where $k$ is a key for the $\pqPRF$ $f$ and $r$ is randomness selected freshly for each encryption. This scheme is characterized (via \expref{Lemma}{lem:SKQES-char}) by the key-independent ``tag state'' $\egoketbra{0^n} \otimes \tau$ (where $\tau$ is the maximally mixed state) 
and the unitary $V_k$ which applies $U_{f_k(\cdot)}$ on the first two registers, controlled on the third register (i.e., the randomness $r$.)

To see why this scheme should be unforgeable, assume for the moment that $U_s$ is a Haar-random unitary and $f_k$ is a perfectly random function. Intuitively, from the point of view of the adversary, each plaintext is mapped into a subspace which is fresh, independent, random, and exponentially-small as a fraction of the total dimension (of the ciphertext space). Security should then reduce to the security of multiple uses of a \QCA one-time scheme, each time with a freshly generated key. We will carefully formalize this intuition in a later section.

\subsubsection{Formal definitions.}

Our definition will compare the performance of an adversary in two games: an unrestricted forgery game, and a cheat-detecting game. Fix an \SKQES $\Pi = (\KeyGen, \Enc, \Dec)$ and let $\adver$ be an adversary in the following.

\begin{experiment}\label{exp:qtest-ctxt}
The $\QUFforge(\Pi, \algo A, n)$ experiment:
	\begin{algorithmic}[1]
		\State $k \from \KeyGen(1^n)$;
		\State \textbf{if} $\Dec_k(\algo A^{\Enc_k}(1^n)) \neq \bot$, \textbf{output} $\win$; otherwise \textbf{output} $\rej$.
	\end{algorithmic}
\end{experiment}
We will think about this experiment as taking place between the adversary $\adver$ and a challenger $\chall$, who generates the key $k$, answers the queries of $\adver$, and then decrypts to see the outcome of the game.  

We now consider a different experiment where $\chall$ attempts to check $\adver$ for cheating. We will make use of the maximally entangled state $\ket{\phi^+}_{M'M''}$ on two copies ($M'$ and $M''$) of the plaintext register, and the corresponding measurement $\{\Pi^+_{M'M''}, \one - \Pi^+_{M'M''}\}$. We will also need a measurement that will help $\chall$ identify previously generated ciphertexts. Recall from~\expref{Section}{sec:characterization} that correctness implies that $\Enc$ can be written in the form $\Enc_k(X)=V_k\bigl(X_M\otimes \sigma_k \bigr)V_k^\dagger$ where $\sigma^{(k)}_T=\sum_r p_k(r) \Pi_{k, r}$ and $\Pi_{k, r} = \proj{\psi^{(k,r)}}_T$. This also defines, for each $(k, r)$, the two-outcome measurement $\{\Pi_{k, r}, \one - \Pi_{k, r}\}$. In all these two-outcome measurements, we denote the first outcome by $0$ and the second outcome by $1$. Notice that these projectors commute, as $\ket{\psi^{(k,r)}}_T$ are elements of an orthonormal basis of eigenvectors.

\begin{experiment}\label{exp:qcheat-ctxt}
	The $\QUFcheat(\Pi, \adver, n)$ experiment:
	\begin{algorithmic}[1]
		\State $\chall$ runs $k \from \KeyGen(1^n)$;
		\State $\adver$ receives $1^n$ and oracle access to $E_k$ (controlled by $\chall$), defined as follows:
		\begin{enumerate}[(1)]
			\item $\adver$ sends plaintext register $M$ to $\chall$;
			\item $\chall$ discards $M$ and prepares $\ket{\phi^+}_{M'M''}$;
			\item $\chall$ applies $\Enc_k$ to $M'$ using fresh randomness $r$, sends result $C$ to $\adver$;
			\item $\chall$ stores $(M'', r)$ in a set $\mathcal M$.
		\end{enumerate}
		\State $\adver$ sends final output register $C_\textsf{out}$ to $\chall$;
		\State $\chall$ applies $V_k^\dagger$ to $C_\textsf{out}$, places results in $MT$; \label{line:apply-unitary}
		\For{\textbf{each} $(M'', r) \in \mathcal M$}
				\State $\chall$ applies $\{\Pi_{k, r}, \one - \Pi_{k, r}\}$  to $T$;\label{line:measure1}
				\If{outcome is $0$}:
					\State $\chall$ applies $\{\Pi^+, \one - \Pi^+\}$ to $MM''$;\label{line:measure2}
					\State \textbf{if} {outcome is $0$}: \textbf{output} $\cheat$\label{line:check-tag-qcheat-ctxt}; \textbf{end if}
			\EndIf
		\EndFor
		\State \textbf{output }$\rej$.
	\end{algorithmic}
\end{experiment}

Note that the experiment always outputs \rej if $\adver$ makes no queries. We emphasize that $\chall$ is a fixed algorithm defined by the security game and the properties of $\Pi$. The challenger is efficient if the states $\proj{\psi^{(k,r)}}$ and the unitary $V_k$ are efficiently implementable and the probability distribution $p_k$ is efficiently sampleable. We believe this is not a significant constraint. It is easily satisfied in all schemes we are aware of. Moreover, in light of \expref{Lemma}{lem:SKQES-char}, it seems unlikely that any reasonable form of ciphertext unforgeability can be defined without this requirement. We are now ready to define security. 

\begin{definition}\label{def:quf-ptxt}
	A $\SKQES$ $\Pi$ has unforgeable ciphertexts (or is $\QUF$) if, for all QPT adversaries $\adver$, it holds:
	$$
	\left|\Pr[\QUFforge(\Pi, \adver, n) \to \win] - \Pr[\QUFcheat(\Pi, \adver, n) \to \cheat]\right| \leq \negl(n)\,.
	$$
\end{definition}

It is straightforward to adapt the above definition to the bounded-query setting, where we fix some positive integer $t$ (at scheme design time) and demand that adversaries can make no more than $t$ queries. We call the resulting notion $\QUF_t$. One then has the obvious implications $\QUF \Rightarrow \QUF_t \Rightarrow \QUF_{t-1} \foral t \in \NN$.

Let us briefly discuss a potential concern with these definitions. Consider the repeated measurements applied to the adversary's final output $C_\textsf{out}$ (\expref{Line}{line:measure1} and \expref{Line}{line:measure2}) in \QUFcheat. The first measurement simply compares the randomness of $C_\textsf{out}$ to that of previously generated ciphertexts. Such measurements will not disturb properly-formed ciphertexts at all, and malformed ones will not affect our security definition. The second measurement actually measures the plaintext register $M$, and thus might (a priori) appear to be concerning. Indeed, if multiple such measurements are applied to $M$, this might open up a vulnerability to attacks. As it turns out, this is not a problem. We will shortly show (see \expref{Theorem}{thm:qUFCTXT-qINDCPA} below) that \QUF implies \QINDCPA. For \QINDCPA schemes, any given random string $r$ is only chosen with negligible probability at encryption time (if not, querying the encryption oracle a polynomial number of times with the challenge plaintext would be enough to compromise security). It follows that, with overwhelming probability, the random strings chosen in the different oracle calls in \QUFcheat are pairwise distinct. This, in turn, implies that the measurement in \expref{Line}{line:measure2} is applied at most once in a given run of the experiment.

\subsubsection{Relationship to other security notions.}\label{sec:comparisonquantum}

It is well-known that even one-time quantum authentication implies \QIND secrecy~\cite{BCG+02}. As we now show, \QUF implies an even stronger notion of secrecy, \QINDCPA. This is a significant departure from classical unforgeability, which is completely independent of secrecy.

\begin{theorem}\label{thm:qUFCTXT-qINDCPA}
If a \SKQES satisfies \QUF, then it also satisfies \QINDCPA.
\end{theorem}
\begin{proof}
Let $\Pi$ be a $\SKQES$, and let $\adver$ be an adversary winning $\QINDCPA$ with non-negligible advantage $\nu$ over guessing, with pre-challenge algorithm $\adver_1$ and post-challenge algorithm $\adver_2$. We will build an adversary $\badver$ with black-box oracle access to $\adver$, able to distinguish between the $\QUFforge$ game and the $\QUFcheat$ game with non-negligible advantage over guessing, as follows:
	\begin{enumerate}
		\item $\badver$ runs $\adver_1(1^n)$, answering its queries using his own oracle $\algo O$;
		\item get registers $M$ (challenge plaintext) and $B$ (side information) from $\adver_1$;
		\item choose a random bit $b \rand \bit$; if $b = 1$, then replace contents of $M$ with a maximally-mixed state;
		\item invoke oracle $\algo O$ on $M$ and place result in register $C$;
		\item run $\adver_2$ on registers $C$ and $B$, receiving output $b' \in \bit$;
		\item if $b = b'$, then output \real; else output \real or \ideal with equal probability.
	\end{enumerate}
Note that, if $\badver$ is playing $\QUFforge$, then $\algo O = \Enc_k$ and we are faithfully simulating the $\QINDCPA$ game for $\adver$. It follows that $b = b'$ with probability at least $1/2 + \nu$. If $\badver$ is playing $\QUFcheat$ instead, $\algo O$ discards its input (and replaces it with half of a maximally-entangled state) on every call. In that case, all inputs to $\adver_1$ and $\adver_2$ are completely uncorrelated with $b$, so that $b'=b$ with probability $1/2$. Therefore, $\adver'$ will correctly guess the game it is playing in with non-negligible advantage.
	
Now it is easy to see how to use $\badver$ to violate the main condition in the definition of $\QUF$ with the same distinguishing advantage. First, query the oracle once and store the output in register $C$. Next, run $\badver$. If $\badver$ outputs \real, then output the contents of $C$ (achieving $\win$ in $\QUFforge$). Otherwise, output a random state in the ciphertext register (achieving $\rej$ in $\QUFcheat$). 
\qed
\end{proof}

We also 
study the restriction of the quantum notion \QUF to the classical case, i.e., classical symmetric-key encryption schemes (\CSKE) vs classical adversaries. We denote this classical restriction by \UF. In this notion, the classical unrestricted forgery game $\UFforge$ is defined precisely as in \expref{Experiment}{exp:qtest-ctxt}. Regarding the quantum game $\QUFcheat$, notice that, in any classical scheme, one can apply ciphertext verification to a string $c$ as follows: (i.) make a copy $c'$ of $c$, (ii.) decrypt $c$, (iii.) if decryption rejected, output reject, and otherwise output $c'$. In other words, all classical encryption schemes automatically satisfy \expref{Condition}{con:efficient}. The appropriate classical restriction $\UFcheat$ of this game thus proceeds as \expref{Experiment}{exp:qcheat-ctxt}, with two modifications: (i.) in step 2: , $\algo C$ replaces the plaintext in register $M_j$ by a random plaintext, encrypts it, and stores a copy of the resulting ciphertext in $C_j$; and (ii.) in step 4: , without decrypting, the game outputs $\cheat$ if the challenge ciphertext $C$ equals any one of the saved $C_j$'s. We then have the following.
\begin{definition}\label{def:uf-ptxt}
A \CSKE $\Pi$ has unforgeable ciphertexts (or is $\UF$) if, for all PPT adversaries $\adver$,
$$
\left|\Pr[\UFforge(\Pi, \adver, n) \to \win] - \Pr[\UFcheat(\Pi, \adver, n) \to \cheat]\right| \leq \negl(n)\,.
$$
\end{definition}

The proof of \expref{Theorem}{thm:qUFCTXT-qINDCPA} carries over easily to the classical case. Moreover, one can show how \UF implies the classical security notion of {\em integrity of ciphertexts} \INTCTXT~\cite{BN00}, which states that no bounded adversary with oracle access to an encryption oracle can produce a ciphertext which is at the same time (i.) valid, and (ii.) fresh, i.e., never output by the oracle. Recall that, classically, it is known~\cite{BN00} that \INTCTXT plus \INDCPA defines authenticated encryption \CAE. Therefore, the notion of unforgeability of ciphertexts, when restricted to the classical case, is at least as strong as authenticated encryption. However, one can also show the converse, i.e., \CAE implies \UF.

\begin{theorem}\label{thm:ae}
$\UF \iff \CAE$.
\end{theorem}
\begin{proof}
The first non-trivial part to prove is $\UF \implies \INTCTXT$. Let $\Pi$ be an \INTCTXT insecure \CSKE. Then there exists an adversary \adver with oracle access to $\Enc_k$ which, with non-negligible probability $\nu$, outputs a ciphertext $c$ which was never output by the encryption oracle. Define a PPT algorithm \badver with oracle access to $\Enc_k$, as follows. First, \badver executes \adver and records a list $L$ of all $\Enc_k$'s answers $c_j$ output to \adver. When \adver outputs a ciphertext $c$, if $c \in L$, \badver outputs a random ciphertext $c'$; else it outputs $c$. For \badver, the success probabilities in the games defining \UF are as follows:
\begin{itemize}
\item in the $\UFforge$ experiment, since $c$ is a fresh ciphertext  with non-negligible probability $\nu$,  \badver wins \UFforge with probability at least $\nu$.
\item In $\UFcheat$ instead, whenever the ciphertext is not fresh, \badver replaces it with a random one, and hence only wins \UFcheat with negligible probability.
\end{itemize}
The fact that a random ciphertext is invalid with overwhelming probability follows by considering an adversary that does not make any queries. So we have:
$$
\left|\Pr[\UFforge(\Pi, \adver', n) \to \win] - \Pr[\UFcheat(\Pi, \adver', n) \to \cheat]\right| \geq  \nu,
$$
and hence $\Pi$ cannot be $\UF$.

The other direction to prove is $\CAE \implies \UF$. For this, we will use an equivalent characterization of \CAE, also known in the literature as $\mathsf{IND\mbox{-}CCA3}$~\cite{Shrimpton04}. In this definition, the adversary's goal is to distinguish whether he's playing in the \CAEreal world, or in the \CAEideal world. In the \CAEreal world, the adversary can interact freely with an encryption oracle $\Enc_k$, and with a restricted decryption oracle $\Dec_k$ which always rejects ($\bot$) decryption queries over any ciphertext which was output by $\Enc_k$. In the \CAEideal world, instead, the adversary is interacting with an oracle $\Enc_k(\$)$ (which ignores the input query, and always returns the encryption of a fresh random plaintext), and a constant $\bot$ oracle (which simulates the decryption oracle but always rejects any query). A scheme $\Pi$ is \CAE secure iff, for any adversary \adver it holds:
$$
\left| \Pr \left[ \CAEreal(\Pi,\adver,n) \to 1 \right] - \Pr \left[ \CAEideal(\Pi,\adver,n) \to 1 \right] \right| \leq \negl(\secpar)\,.
$$
Now, let \adver be a PPT adversary breaking \UF for a scheme $\Pi$. This means that there exists a non-negligible function $\nu$ such that:
$$
\left|\Pr[\UFforge(\Pi, \adver, n) \to \win] - \Pr[\UFcheat(\Pi, \adver, n) \to \cheat]\right| \geq \nu(n)\,.
$$
We use \adver to build an adversary \badver able to distinguish \CAEreal from \CAEideal. The new adversary \badver runs \adver and forwards all of \adver's encryption queries to his own encryption oracle. Finally, when \adver outputs a ciphertext $c$, \badver queries his own decryption oracle on $c$, and looks at the oracle's response. If the response is {\em not} $\bot$, then \badver returns \real, otherwise returns \real or \ideal with equal chance.

It is easy to see that \badver distinguishes \CAEideal from \CAEreal with non-negligible advantage at least $\nu/2$ over guessing. The reason is as follows. If \badver is in the \CAEreal world (probability $1/2$), then he is correctly simulating for $\adver$ the \UFforge game. Since \adver breaks \UF by assumption, it means that, with probability at least $\nu$, his output $c$ will be a fresh valid ciphertext; in that case, also \badver wins. On the other hand, if the world is \CAEideal, \badver still wins with probability $1/2$.\qed
\end{proof}

This means that \UF is actually {\em another characterization of authenticated encryption}. This is an interesting observation, given that \UF comes from the classical restriction of a quantum notion ``merely'' concerning the unforgeability of ciphertexts. However, we stress that this equivalence only holds at the classical level, and that this is insufficient evidence to declare that \QUF serves the same purpose quantumly as \CAE does classically. In fact, in \expref{Section}{sec:QAE} we introduce a quantum analogue of \CAE which we call \QAE, and provide stronger evidence that the latter is in fact the correct analogue.

\section{Quantum \INDCCAA}

Next, we move to the problem of defining adaptive chosen-ciphertext security for quantum encryption. In the usual classical formulation (\INDCCAA), the adversary $\algo A$ receives both an encryption oracle and a decryption oracle for the entire duration of the indistinguishability game. To eliminate the trivial strategy, we do not permit $\algo A$ to query the decryption oracle on the challenge ciphertext. This last condition does not make sense in the quantum setting, for a number of reasons we've seen before: no-cloning prevents us from storing a copy of the challenge, measurement may destroy the states involved, and so on. However, our approach to defining unforgeability can be adapted to this case. The resulting notion of {\em quantum indistinguishability under adaptive chosen-ciphertext attacks} (\QINDCCAA) can also be recast 
in the public-key quantum encryption setting.

\subsubsection{Formal Definition.}

As before, we will compare the performance of the adversary in two games. In each case, the adversary $\adver=(\algo A_1,\algo A_2)$ consists of two parts (pre-challenge and post-challenge), and is playing against the challenger $\chall$, which is a fixed algorithm determined only by the security game and the scheme.

\begin{experiment}\label{exp:INDCCAA-test}
The $\QINDCCAAtest(\Pi, \algo A, n)$ experiment:
	\begin{algorithmic}[1]
		\State $\chall$ runs $k \from \KeyGen(1^n)$ and flips a coin $b \inrand \bit$;
		\State $\algo A_1$ receives $1^n$ and access to oracles $\Enc_k$ and $\Dec_k$;
		\State $\algo A_1$ prepares a side register $S$, and sends $\chall$ a challenge register $M$;
		\State $\chall$ puts into $C$ either $\Enc_k(M)$ (if $b=0$) or $\Enc_k(\tau_M)$ (if $b=1$);
		\State $\algo A_2$ receives registers $C$ and $S$ and oracles $\Enc_k$ and $\Dec_k$;
		\State $\algo A_2$ outputs a bit $b'$.  \textbf{If} $b'=b$, \textbf{output} \textsf{win}; otherwise \textbf{output} \textsf{fail}.
	\end{algorithmic}
\end{experiment}

Notice that in this game there are no restrictions on the use of $\Dec_k$ by $\algo A_2$. In particular, $\algo A_2$ is free to decrypt the challenge. In the second game, the challenge plaintext is replaced by half of a maximally entangled state, and $\adver$ only gains an advantage over guessing if he cheats, i.e., if he tries to decrypt the challenge.

\begin{experiment}\label{exp:INDCCAA-fake}
The $\QINDCCAAfake(\Pi, \algo A, n)$ experiment:
	\begin{algorithmic}[1]
		\State $\chall$ runs $k \from \KeyGen(1^n)$;
		\State $\algo A_1$ receives $1^n$ and access to oracles $\Enc_k$ and $\Dec_k$;
		\State $\algo A_1$ prepares a side register $S$, and sends $\chall$ a challenge register $M$;
		\State $\chall$ discards $M$, prepares $\ket{\phi^+}_{M'M''}$ and fresh randomness $r$, and stores $(M'', r)$; then $\chall$ encrypts the $M'$ register and sends the resulting ciphertext $C'$ to $\algo A_2$;
		\State $\algo A_2$ receives registers $C'$ and $S$ and oracles $\Enc_k$ and $D_k$, where $D_k$ is defined as follows. On input a register $C$:
			\begin{enumerate}[(1)]
			\item $\chall$ applies $V_k^\dagger$ to $C$, places results in $MT$; 
			\item $\chall$ applies $\{P_T^{\sigma_k}, \one - P_T^{\sigma_k}\}$  to $T$; 
				\If{outcome is $0$}:
					\State $\chall$ applies $\{\Pi_{k, r}, \one - \Pi_{k, r}\}$  to $T$;  
					\If{outcome is $0$}: 
						\State $\chall$ applies $\{\Pi^+, \one - \Pi^+\}$ to $M M''$;
						\State \textbf{if} {outcome is $0$}: \textbf{output} \textsf{cheat};
					\EndIf
				\Else 
					\State apply the default map for invalid ciphertexts, i.e., $\hat D_k$ to $MT$. 
				\EndIf
		\State	\textbf{return} $M$;
		\end{enumerate}
		\State $\chall$ draws a bit $b$ at random. \textbf{If} $b=1$, \textbf{output} \textsf{cheat}; if $b=0$ \textbf{output} \textsf{reject}.
	\end{algorithmic}

\end{experiment}

We now define quantum \INDCCAA in terms of the advantage gap of adversaries between the above two games.\footnote{The interface that the two games provide to the adversary differ slightly in that the adversary is not asked to output a bit in the end of the \QINDCCAAfake game. This is not a problem as the games have the same interface until the second one terminates.}

\begin{definition}\label{def:qINDCCAA}
	A $\SKQES$ $\Pi$ is $\QINDCCAA$ if, for all QPT adversaries $\adver$,
	$$
	\Pr[\QINDCCAAtest(\Pi, \adver, n) \to \win] - \Pr[\QINDCCAAfake(\Pi, \adver, n) \to \cheat] \leq \negl(n)\,.
	$$
\end{definition}

The omission of absolute values in the above is intentional. Indeed, an adversary can artificially inflate his cheating probability by querying the decryption oracle on the challenge and then ignoring the result. What he should not be able to do (against a secure scheme) is make his win probability larger than his cheating probability. We note that \QINDCCAA clearly implies \QINDCCA.
\begin{proposition}
	$\QINDCCAA\implies\QINDCCA$.
\end{proposition}
\begin{proof}
Suppose we have a scheme $\Pi$ which is not $\QINDCCA$, i.e., there exists an adversary $\adver$ which wins the usual \QINDCCA game with non-negligible advantage $\nu$ over guessing. 
Clearly $\adver$ can also play the games \QINDCCAAtest and \QINDCCAAfake, but will not query the decryption oracle post-challenge. Note that $\adver$ wins \QINDCCAAtest with probability $1/2 + \nu$, 
but is declared as cheating in \QINDCCAAfake with probability exactly $1/2$. Hence $\Pi$ is not $\QINDCCAA$.\qed
\end{proof}

Next, we show that the classical restriction of \QINDCCAA is equivalent to the classical security notion \INDCCAA. We denote the classical restriction of $\QINDCCAA$ by $\INDCCAA'$. This is defined by adapting the replacement and verification procedure of the challenger in \QINDCCAAtest in the same way as when defining \UF. We denote the classical versions of the games \QINDCCAAtest and \QINDCCAAfake by \INDCCAAtest and \INDCCAAfake, respectively. 
\begin{theorem}
	A $\SKES$ $\Pi$ is $\INDCCAA'$ iff it is $\INDCCAA$.
\end{theorem}
\begin{proof}
	Suppose first that $\adver$ is an adversary breaking \INDCCAA', i.e., winning \INDCCAAtest with a probability higher than the one winning \INDCCAAfake by a 
	non-negligible advantage $\nu$. We construct an adversary $\adver'$, that runs $\adver$, keeps a copy of the challenge ciphertext and aborts by giving a random answer whenever $\adver$ is about to query the decryption oracle with the challenge ciphertext. Note that $\adver'$ wins \INDCCAAfake with probability exactly $1/2$. We call $\adver'$ the {\em self-checking version of $\adver$}. 
It is easy to show 
that $\adver'$ wins the \INDCCAAtest game with probability at least $1/2+\nu$. First observe that the probability that $\adver$ cheats is the same in \INDCCAAtest and \INDCCAAfake. This is because the two games are identical up to the point where $\adver$ sends their first cheating query. Moreover we have 
\begin{align}
	\prob\left[\adver \text{ wins } \INDCCAAtest \wedge \adver \text{ cheats}\right]&\le \prob\left[ \adver \text{ cheats}\right]\nonumber\\
	&= \prob\left[ \adver \text{ wins } \INDCCAAfake \wedge\adver \text{ cheats}\right],\nonumber
\end{align}
implying
\begin{align}
	&\prob\left[\adver \text{ wins } \INDCCAAtest \wedge \adver \text{ does not cheat}\right]\nonumber\\=&\prob\left[\adver \text{ wins } \INDCCAAtest\right]-\prob\left[\adver \text{ wins } \INDCCAAtest \wedge \adver \text{ cheats}\right]\nonumber\\
	\ge& \prob\left[\adver \text{ wins } \INDCCAAfake\right]-\prob\left[\adver \text{ wins } \INDCCAAfake \wedge \adver \text{ cheats}\right]+\nu\nonumber\\
	=&\prob\left[\adver \text{ wins } \INDCCAAfake \wedge \adver \text{ does not cheat}\right]+\nu\nonumber\\
	=&\frac 1 2 \prob\left[\adver \text{ does not cheat}\right]+\nu.\nonumber
\end{align}
It follows that
	\begin{align}
		&\prob\left[\adver'\text{ wins }\INDCCAAtest\right]\nonumber\\
		=\,&\prob\left[\adver\text{ wins }\INDCCAAtest\wedge \adver\text{ does not cheat}\right]+\frac 1 2\prob\left[\adver \text{ cheats}\right]\nonumber\\
		\geq\,&\frac 1 2\prob\left[ \adver\text{ does not cheat}\right]+\nu+\frac 1 2\prob\left[\adver \text{ cheats}\right]\nonumber\\
		= &\frac{1}{2}+\nu\,.\nonumber
	\end{align}
But the \INDCCAAtest and \INDCCAA games are identical for adversaries that do not query the challenge, and $\adver'$ has been constructed not to, i.e., $\adver'$ wins the \INDCCAA game with probability $1/2+\nu
$.

For the other direction, let $\adver$ be an adversary that wins the \INDCCAA game with 
non-negligible advantage.
Note that $\adver$ 
behaves the same in
all games, as any difference only arrises upon cheating, and \adver does not cheat by definition of the \INDCCAA game. Therefore \adver wins the $\INDCCAAtest$ game with 
non-negligible advantage over random guessing by assumption, but it wins the \INDCCAAfake game with probability exactly $\frac 1 2$.\qed
\end{proof}

\section{Quantum Authenticated Encryption}\label{sec:QAE}

In the classical setting, authenticated encryption (\CAE) is defined as \INDCCAA and unforgeability of ciphertexts (see Definition 4.17 in~\cite{KL14}) or, equivalently, \INDCPA and unforgeability of ciphertexts~\cite{BN00}. A third equivalent formulation due to Shrimpton~\cite{Shrimpton04} defines \CAE in terms of a real vs ideal scenario. According to this definition, a classical scheme $\Pi = (\KeyGen, \Enc, \Dec)$ is \CAE if no adversary, given oracles $E$ and $D$, can distinguish these two scenarios:
\begin{itemize}
\item \CAEreal : $(E, D)$ is $(\Enc_k, \Dec_k)$ with $k \from \KeyGen$;
\item \CAEideal : $E$ discards the input and returns $\Enc_k(m)$ for random $m$, and $D$ always rejects; here again $k \from \KeyGen$; 
\end{itemize}
This is not yet enough, because the adversary $\adver$ can always distinguish real from ideal by composing $E$ with $D$. To patch this problem, we can (i.) demand that $\adver$ cannot do that, as in~\cite{Shrimpton04}, or (ii.) add the condition $D \circ E = \one$ to the ideal case\footnote{More precisely, the ideal world maintains a list of all queries that $\adver$ makes to $E$, and ensures that $D$ will respond correctly if queried on an output of $E$.}. We will take the latter approach.

Motivated by this formulation of \CAE and our general strategy so far, we will define quantum authenticated encryption by comparing the performance of the adversary in a real world and an ideal world. In the real world, the adversary gets unrestricted access to $\Enc_k$ and $\Dec_k$. In the ideal world, the challenger $\chall$ stores the $\Enc_k$ queries, replacing them with halves of maximally-entangled states; when a $\Dec_k$ query is detected as corresponding to a particular earlier $\Enc_k$ query, $\chall$ replies with the contents of the stored register; otherwise $\Dec_k$ rejects. Cheat detection is performed just as in the unforgeability game \QUFcheat.

\subsubsection{Formal definition.}

We now formally define the two worlds: the real world \QAEreal, and the ideal (or cheat-detecting) world \QAEideal. In both cases, the adversary $\algo A$ receives two oracles and then outputs a single bit. 

\begin{experiment}\label{exp:QAErealeal}
The $\QAEreal(\Pi, \algo A, n)$ experiment:
	\begin{algorithmic}[1]
		\State $k \from \KeyGen(1^n)$;
		\State \textbf{output} $\algo A^{\Enc_k, \Dec_k}(1^n)$.
	\end{algorithmic}
\end{experiment}

In the ideal setting, it will be convenient to describe the experiment in terms of an interaction between $\adver$ and the challenger $\chall$, a fixed algorithm determined only by the security game and the properties of $\Pi$.

\begin{experiment}\label{exp:QAEideal}
The $\QAEideal(\Pi, \algo A, n)$ experiment:
	\begin{algorithmic}[1]
		\State $\chall$ runs $k \from \KeyGen(1^n)$;
		\State initialize 	oracles $E_{M \rightarrow C}$ and $D_{C \rightarrow M}$:
		\begin{itemize}
		\item $E$ is defined as follows. On input a register $M$:
			\begin{enumerate}[(1)]
			\item $\chall$ prepares $\ket{\phi^+}_{M'M''}$, and generates fresh randomness $r$;
			\item $\chall$ stores $(r, M'', M)$ in a set $\mathcal M$;
			\item $\chall$ applies $\Enc_k$ to $M'$ using randomness $r$; \textbf{return} result to $\adver$.
			\end{enumerate}
		\item $D$ is defined as follows. On input a register $C$:
			\begin{enumerate}[(1)]
			\item $\chall$ applies $V_k^\dagger$ to $C$, places results in $M'T$; 
			\For {\textbf{each} $(r, M'', M) \in \mathcal M$}: 
				\State $\chall$ applies $\{\Pi_{k, r}, \one - \Pi_{k, r}\}$  to $T$;
				\If{outcome is $0$}:
					\State $\chall$ applies $\{\Pi^+, \one - \Pi^+\}$ to $M'M''$;
					\State \textbf{if} {outcome is $0$}: \textbf{return} $M$\label{line:check-tag};
			\EndIf
		\EndFor
		\State	\textbf{return} $\egoketbra{\bot}$;
		\end{enumerate}
		\end{itemize}
		\State \textbf{output} $\adver^{E,D}(1^n)$.
	\end{algorithmic}
\end{experiment}
Note that, as before, we number the measurement outcomes by $0$ (the first outcome) and $1$ (the second outcome). With the above games defined, we can now set down our definition of quantum authenticated encryption.

\begin{definition}
A \SKQES $\Pi=(\KeyGen,\Enc,\Dec)$ is an {\em authenticated quantum encryption scheme} (or is \QAE) if, for all QPT adversaries $\adver$:
 	\begin{equation}
 	\left|\Pr\left[\QAEreal(\Pi, \algo A, n) \to \real \right] -
 	\Pr\left[\QAEideal(\Pi, \algo A, n) \to \real \right]\right| \leq \negl(n).\nonumber
 	\end{equation}
\end{definition}

\subsubsection{Relationship to other security notions.}

Next, we give evidence that \QAE is indeed the correct formalization of a quantum analogue of \CAE, by showing that it implies all of the quantum security notions defined thus far. We begin with adaptive chosen-ciphertext security.

\begin{theorem}\label{thm:QAECCA2}
	$\QAE \implies \QINDCCAA$.
\end{theorem}
\begin{proof}
The proof is similar to that of Theorem~\ref{thm:qUFCTXT-qINDCPA}. 
For a scheme $\Pi$, let \adver be an adversary against \QINDCCAA, e.g., let us say that:
$$
\Pr\left[ \QINDCCAAtest(\Pi,\adver,n) \to \win \right] = \Pr\left[ \QINDCCAAfake(\Pi,\adver,n) \to \cheat \right] + \nu(n)\, ,
$$
for non-negligible $\nu$. We then show how to build another adversary $\badver$ with black-box access to \adver, able to distinguish \QAEreal from \QAEideal.

$\badver$ runs \adver, and forwards all of \adver's queries to his own oracles. When eventually \adver outputs a challenge plaintext state, \badver flips a random bit $b$. If $b=0$, then \badver forwards the challenge plaintext to his encryption oracle as usual. Otherwise, if $b=1$, \badver replaces the challenge with a totally mixed plaintext state before relaying it to the oracle. After that, \badver continues to answer \adver's queries during the second quantum CCA phase as before, by forwarding all the queries to his oracles, until \adver produces an output bit $b'$. Finally, if $b=b'$, then \badver outputs \real, otherwise he outputs \ideal.

Now notice the following: If we are in the \QAEreal environment (that is, \badver has unrestricted \Enc and \Dec oracles), then \badver is faithfully simulating for \adver the \QINDCCAAtest game, which means that the probability of \badver correctly outputting \real is exactly the same probability of \adver of winning \QINDCCAAtest.

If we are in the \QAEideal world, instead, \badver is playing in a ``malformed'' game, where all his encryption queries are replaced by random plaintexts before encryption. This means that the best \adver could do in order to guess the secret bit $b$ is guessing at random, {\em unless} \adver uses a ``cheating decryption query'' on the challenge ciphertext (in this case the modified decryption oracle of the game \QAEideal would actually return the encrypted plaintext). It follows that 
	\begin{align}
		& \Big| \Pr\left[ \QAEreal(\Pi,\badver,n) \to \mathsf{Real} \right] - \Pr\left[ \QAEideal(\Pi,\badver,n) \to \mathsf{Real} \right] \Big| \nonumber\\
		\ge & \Big| \Pr\left[ \QINDCCAAtest(\Pi,\adver,n) \to \win \right] - \Pr\left[ \QINDCCAAfake(\Pi,\adver,n) \to \cheat \right] \Big| \\
		=&\Pr\left[ \QINDCCAAtest(\Pi,\adver,n) \to \win \right] - \Pr\left[ \QINDCCAAfake(\Pi,\adver,n) \to \cheat \right]\nonumber
		=\nu \, ,\nonumber
	\end{align}
which conludes the proof.\qed
\end{proof}

In terms of authentication security, we can show that \QAE implies \cQCA (computational one-time ciphertext authentication), and hence also \cDNS.

\begin{theorem}
	Let $\Pi=(\KeyGen, \Enc, \Dec)$ be a \SKQES that is \QAE secure and satisfies \expref{Condition}{con:efficient}. Then it is \cQCA.
\end{theorem}
\begin{proof}
Assume $\Pi$ is not \cQCA. Then there exists an algorithm $\adver=(\adver_1,\adver_2, \adver_3)$ that achieves the following. $\adver_1$ gets an input $1^n$ and outputs registers $M$ (the plaintext register) and $B$. $\adver_2$ implements a map $\Lambda_{CB\to C\tilde B}$ on two registers $C$ (the ciphertext register) and $B$. $\adver_3$ is a distinguisher between the two states resulting from applying $ \tilde\Lambda_{CB\to C\tilde B}$ or the corresponding simulator according to Equations \eqref{eq:simulator1} and \eqref{eq:simulator2} to the output of $\adver_1$.
	
The crucial observation is, that the map on registers $MB$ resulting from sending $M$ to the challenger $\chall'_{\ideal}$ as an encryption query in the ideal \QAE game, applying $\Lambda_{CB\to C\tilde B}$ to the output and sending the resulting $C$-register to $\chall'_{\ideal}$ as a decryption query, is exactly the simulator defined in Equations~\eqref{eq:simulator1} and~\eqref{eq:simulator2}. Thus, the adversary that runs $\adver_1$, queries the encryption oracle, runs $\adver_2$, queries the decryption oracle and finally runs $\adver_3$ is a successful \QAE adversary.\qed
\end{proof}

In addition, \QAE implies quantum unforgeability.

\begin{theorem}\label{thm:QAEUF}
	$\QAE \implies \QUF$.
\end{theorem}
\begin{proof}
For a scheme $\Pi$, let \adver be an adversary against \QUF, e.g., let us say that:
$$
\Pr\left[ \QUFforge(\Pi,\adver,n) \to \win \right] = \Pr\left[\QUFcheat(\Pi,\adver,n) \to \cheat \right] + \nu \, ,
$$
where $\nu$ is non-negligible. We then build another adversary $\badver$ with black-box access to \adver, able to distinguish \QAEreal from \QAEideal with non-negligible advantage. $\badver$ runs \adver, and forwards all of \adver's queries to his own encryption oracle. When eventually \adver outputs a forgery, $\badver$ sends it for decryption to his own decryption oracle. If the decryption succeeds (that is, the oracle does not return $\egoketbra{\bot}$), then $\badver$ outputs \real, otherwise he outputs \ideal.

The idea is the following: suppose the decryption of the forgery state succeeds (i.e., it does not decrypt to $\egoketbra{\bot}$). This can happen in two cases:
\begin{enumerate}
\item we are in the \QAEreal game, and \adver produced a valid forgery (i.e., he won the \QUFforge game); or
\item we are in the \QAEideal game, and \adver cheated by replaying an output of the encryption oracle (i.e., he won the \QUFcheat game).
\end{enumerate}
Recall that, by assumption, \adver produces a valid forgery with probability at least $\nu$ over cheating. Therefore the case 2. above happens with noticeable less probability than case 1., which is in fact the one \badver ``bets'' on. Analogously, suppose the decryption fails. This can happen in two cases:
\begin{enumerate}
\item we are in the \QAEreal game, but \adver produced an invalid forgery (i.e., he lost the \QUFforge game); or
\item we are in the \QAEideal game, and \adver did not cheat (i.e., he lost \QUFcheat).
\end{enumerate}
For the same reasoning as above, 2. is noticeably more likely than 1., which is in fact \badver's bet. More in detail, we have:
\begin{align}
	& \Big| \Pr\left[ \badver(\QAEreal) \to \mathsf{Real} \right] - \Pr\left[ \badver(\QAEideal) \to \mathsf{Real} \right] \Big| \nonumber\\
	=\,& \Big| \Pr\left[ \QAEreal \right] \cdot \Pr\left[ \adver(\QUFforge) \to \win \right] - \nonumber\\
	~~~-\,& \Pr\left[ \QAEideal \right] \cdot \Pr\left[ \adver(\QUFcheat) \to \cheat \right] \Big| \nonumber\\
	=\,& \frac{1}{2}\Big| \Pr\left[ \adver(\QUFforge) \to \win \right] - \big( \Pr\left[ \adver(\QUFforge) \to \win \right] - \nu \big) \Big|=\frac{\nu}{2} \, ,\nonumber
\end{align}
which is non-negligible. \qed
\end{proof}

Finally, we consider the classical restriction $\CAE'$ of \QAE.

\begin{proposition}\label{prop:CAECAE}
	$\CAE' \iff \CAE$.
\end{proposition}
\begin{proof}
The security notion \CAE' is given in terms of two experiments which are like the \CAEreal and \CAEideal experiments in Shrimpton's formulation of \CAE security, with the following difference:
\begin{enumerate}
\item in the modified \CAEreal experiment, the decryption oracle does not reject non-fresh ciphertexts, i.e. it is unrestricted; and
\item in the modified \CAEideal experiment, the decryption oracle does not always return $\bot$: in case it is queried on a non-fresh ciphertext, it decrypts correctly.
\end{enumerate}
Since classically we can store and compare plaintexts and ciphertexts, it is easy to construct an efficient simulator able to switch between the experiments of \CAE and \CAE', by inspecting \adver's decryption queries and reacting accordingly. Namely:
\begin{enumerate}
\item to switch from \CAE to \CAE', record \adver's plaintexts and ciphertexts during encryption queries, and reply with the right plaintext whenever \adver asks to decrypt a non-fresh ciphertext (otherwise, just send the query to the decryption oracle); and
\item to switch from \CAE' to \CAE, record \adver's received ciphertexts during encryption queries, and reply with $\bot$ whenever \adver asks to decrypt a non-fresh ciphertext  (otherwise, just send the query to the decryption oracle).
\end{enumerate}
This concludes the proof, as it shows the two cases to be equivalent.\qed
\end{proof}

In particular, $\CAE'$ is equivalent to $\UF$. We provide evidence that a quantum analogue of this statement does not hold in the next section.

\section{Constructions and separations}

In this section we exhibit constructions of \SKQES that fulfill and separate the different security notions presented in the preceding sections. We begin by showing that augmenting a one-time scheme by a (perfectly) random function family using the construction in \expref{Definition}{def:SKE-generic} turns a \QCA secure scheme into a \QAE secure scheme. Then we will move on to show how to satisfy \QAE with an efficiently implementable scheme. Recall that efficient $\QCA$-secure \SKQES can be constructed, e.g., from unitary two-designs like the Clifford group.
\begin{theorem}\label{thm:QAErandom}
	Let $\Pi$ be a $\QCA$-secure \SKQES, and let $f: \K \times \bit^n \to \bit^m$ be a random function family. Then the scheme $\Pi^f$ in \expref{Definition}{def:SKE-generic} is \QAE secure.
\end{theorem}
\begin{proof}

		We let $\Pi=(\KeyGen, \Enc,\Dec)$ and $\Pi^\algo F =(\KeyGen', \Enc',\Dec')$ where
	\begin{enumerate}
		\item $\KeyGen'(1^n)$ outputs a random function $F$ from $\bit^n$ to $\bit^m$;
		\item $\Enc'_F(X_M)$ outputs $\proj{s}_R\otimes \Enc_{F(s)}(X)_C$ , where $s \rand \bin^n$;
		\item $\Dec'_F(Y_{RC})$ first measures the $R$ register to get outcome $s'$; then it runs $\Dec_{F(s')}$ on register $C$ and outputs the result.
	\end{enumerate}
	Suppose $\adver$ is a \QAE adversary against $\Pi^\algo F$, i.e., a QPT algorithm with oracle access to $\Enc'_k$ and $\Dec'_k$. Suppose $\adver$ makes $\ell(n)$ queries to the oracle, where $\ell$ is some polynomial function of $n$. We assume that the randomnesses $s_i$ and the keys $F(s_i)$ used for the scheme $\Pi$ in the different encryption queries (for $i = 1, \ldots, \ell(n)$) are all distinct; this is true except with negligible probability. 
	
	Let us first analyze what happens in the \QAEreal experiment. Consider the $i$-th decryption oracle call. The decryption begins with a measurement of the $R$ register, yielding some outcome $s$ and thereby a key $\bar k=F(s)$. We can analyze the situation for each outcome $s$ that occurs with non-negligible probability, separately. This is because if an adversary is successful, it is easy to see that there is also a modified successful adversary, that submits only decryption queries with a fixed string $s$ in the randomness register.
	
	Suppose first that $\bar k=F(s) \neq F(s_i)$ for all $i$. In this case, the $\Pi$-encrypted part of the forgery candidate gets decrypted with a key different from all the ones used for encryption. We analyze the attack map $\Lambda=\tilde\adver(1^n)\tr_C$ against the \QCA scheme $\Pi$, where $\tilde\adver$ is defined to first run $\adver$ until the $i$th decryption query, while answering each encryption query by sampling a fresh key for the scheme $\Pi$. Note that $\Lambda$ does not use  initial side information, therefore $\sigma^{\acc}:=\Lambda^{\acc}$ and $\sigma^{\crej}:=\Lambda^{\crej}$ are just positive semidefinite matrices whose trace sums to one.

	According to Equation \eqref{eq:extra-constraint} in the definition of \QCA, the trace of $\sigma^{\acc}$ is the probability that the simulator applies the identity to the plaintext. The output of the attack map $\Lambda$ does not depend on it's input, i.e. the same holds for the effective map $\Lambda^\Pi$ and hence for $(\mathds{1}-\proj{\bot})\Lambda^\Pi(\cdot)(\mathds{1}-\proj{\bot})$. Any such map is far from any non-negligible multiple of the identity channel so the trace of $\sigma^{\acc}$ is negligible according to Equation \ref{eq:normal-constraint}. We have hence shown that the decryption oracle returns $\bot$ with ovewhelming probability, so we can take $\sigma^{crej}=\tr_C\tilde\adver(1^n)$.

	Let now $s'=r_j$, and write $\adver=\adver_1\Enc_{\hat k}\adver_0$, splitting the adversary into two parts before and after the $j$-th encryption query. Let $(\tilde \adver_1)_{CE_1\to CE_2}$ be defined analogous to $\tilde \adver$. $E_1$ and $E_2$ are the internal memory registers of $\adver$ at the time of the $j$-th encryption query and the $i$-th decryption query, respectively. $\Pi$ is \QCA secure, implying that $\tilde \adver_1^\Pi=\mathbb{E}_{\bar k}\left[\Dec_k\circ\tilde \adver_1\circ\Enc_{\bar k}\right]$ fulfills:
	\begin{equation}\label{eq:sim1}
	\|(\tilde \adver^\Pi_1)_{ME_1\to ME_2}-\id_M\otimes(\tilde \adver^{\acc}_1)_{E_1\to E_2}-\bot\otimes(\tilde \adver^{\crej}_1)_{E_1\to E_2}\|_\diamond\le\negl(n),
	\end{equation}
	where (using $P_{\text{inv}}=\mathds 1-\proj{\Phi_{\bar k,\bar r}}$):
	\begin{align}\label{eq:sim2}
	\tilde \adver^{\acc}_1 &= \mathbb{E}_{\bar k,\bar r}\left[\bra{\Phi_{\bar k,\bar r}} V_{\bar k}^\dagger\tilde \adver^{\acc}_1 \left(\Enc_{\bar k;\bar r}\left(\phi^+_{MM'} \right)\otimes (\cdot)_{E_1} \right)V_{\bar k}\ket{\Phi_{\bar k, \bar r}}\right]\text{ and}\nonumber\\
	\tilde \adver^{\crej}_1 &= \mathbb{E}_{\bar k,\bar r}\left[\tr_{MM'T}P_{\text{inv}} V_{\bar k}^\dagger\tilde \adver^{\acc}_1 \left(\Enc_{\bar k;\bar r}\left(\phi^+_{MM'} \right)\otimes (\cdot)_{E_1} \right)V_{\bar k}\right].
	\end{align}
	The form of the simulator in the reject case follows by using that the maximally entangled state is a point in the optimization defining the diamond norm in \eqref{eq:normal-constraint} and using the monotonicity of the trace norm under partial trace.

We now show indistinguishability of the real and ideal experiments by induction over the decryption queries. Since \QCA implies \IND, the two are indistinguishable before the first decryption query. Assume now that the two experiments cannot be distinguished using an algorithm that makes at most $i-1$ decryption queries.
Consider \adver running in the ideal experiment until right before the $(i+1)$-th decryption query (or until the end, if $i=\ell$). We make the same case distinction as before. In the first case the measurement in line (3) in the ideal decryption oracle in \expref{Experiment}{exp:QAErealeal} never returns 0, i.e. the output is always \rej. Therefore we can replace the $i$-th decryption oracle by the constant reject function, thereby reducing the number of decryption oracle calls of to $i-1$. By the induction hypothesis, the contents of the internal register are therefore indistinguishable whether in the \QAEreal or in the \QAEideal experiment.
	
	Turning to the second case, we make a very similar argument. We have $s=s_j$, i.e. the only encryption query where the measurement from line (3) in the definition of the ideal decryption oracle in \expref{Experiment}{exp:QAErealeal} can possibly return $0$ is the $j$-th. Here it is left to observe that the rest of the ideal decryption oracle implements exactly the same map as in the ideal world, i.e. the ones from equations \eqref{eq:sim1} and \eqref{eq:sim2}. Replacing the $j$-th encryption and the $i$-th decryption oracle call by this map, and using the induction hypothesis, we get that \adver run until before the $i+1$-th decrytion oracle call cannot distinguish \QAEreal from \QAEideal. This ends the proof by induction. \qed

\end{proof}

We now show how to satisfy \QAE efficiently, by means of a post-quantum-secure pseudorandom function.

\begin{corollary}\label{cor:QAEconstr}
	Let $\Pi$ be a $\QCA$-secure \SKQES that satisfies \expref{Condition}{con:efficient}, and let $f$ be a \pqPRF. Then the scheme $\Pi^f$ (from \expref{Definition}{def:SKE-generic}) satisfies \QAE.
\end{corollary}
\begin{proof}
	As a contradiction, suppose there exists a QPT algorithm $\adver$ that distinguishes \QAEreal from \QAEideal. We claim that this also holds if $f$ is replaced with a completely random function family $\algo F$. If $\adver$ cannot break the random scheme $\Pi^{\algo F}$, then we can build a distinguisher for $f$ versus $\algo F$, as follows. What we would like to do is the following. Given an oracle $\algo O$, we:
	\begin{enumerate}
		\item choose a random bit $b \inrand \bit$;
		\item if $b=0$, we simulate the $\QAEreal(\Pi^\algo O, \algo A, n)$ experiment using our oracle;
		\item if $b=1$, we simulate the $\QAEideal(\Pi^\algo O, \algo A, n)$ experiment using our oracle;
		\item output $b \oplus s$ where $s$ is the output of \adver.
	\end{enumerate}
	This may at first not seem possible using the classical oracle we are provided with, as the ideal decryption oracle has to implement the unitary $V_k^\dagger$, which seems to require superposition access to the random/pseudorandom function. However, observe that steps 5-11 of Experiment \ref{exp:qcheat-ctxt} commute with a measurement of the randomness register $R$ in the computational basis, and afterwards this register is discarded. Therefore the outcome of the experiment is not changed by first measuring the register $R$, which yields an outcome $r$. Then the modified challenger can use classical oracle access to the random/pseudorandom function to implement $V_k^\dagger$ on the measured input state.
	
	Note that, if $\Pi^\algo O$ is secure, then $b$ and $s$ are independent (up to negligible terms) and $b \oplus s$ is a fair coin. If $\Pi^\algo O$ is insecure, then it deviates from uniform by the $\QUF$ distinguishing advantage of $\algo A$. This yields a distinguisher between the case $\algo O = f$ and $\algo O = \algo F$. The claim then follows from \expref{Theorem}{thm:QAErandom}.
	\qed
\end{proof}

In particular, the scheme family $\des^\pqPRF$ is sufficient for \QAE. We remark that the proof uses the fact that, given classical oracle access to $f$, the scheme $\Pi^f$ is efficiently implementable in the sense of \expref{Condition}{con:efficient} -- regardless of the nature of the family $f$. Of course, in the special case where $f$ is a \pqPRF, then $\Pi^f$ simply satisfies \expref{Condition}{con:efficient} without any need for oracles.

As \QAE implies both \QUF and \QINDCCAA (see \expref{Theorem}{thm:QAEUF} and \expref{Theorem}{thm:QAECCA2}), we have the following corollary.

\begin{corollary}
	Let $\Pi$ be a $\QCA$-secure \SKQES that satisfies \expref{Condition}{con:efficient}, and let $f$ be a \pqPRF. Then the scheme $\Pi^f$ (from \expref{Definition}{def:SKE-generic}) satisfies \QUF and \QINDCCAA.
\end{corollary}

We can also show how to satisfy bounded-query unforgeability, i.e., $\QUF_t$. Recall that a $t$-wise independent function is a deterministic, efficiently computable keyed function family $\{f_k\}_k$ which appears random to any algorithm (of unbounded computational power) which gets classical oracle access to $f_k$ for uniformly random $k$, and can make at most $t$ queries. One can apply the proof technique of \expref{Corollary}{cor:QAEconstr} and Theorem \ref{thm:QAErandom} to obtain the following.

\begin{corollary}\label{cor:t-time}
Let $\Pi$ be a $\QCA$-secure \SKQES, and let $f$ be a $t$-wise independent function family. Then the scheme $\Pi^f$ (as defined in \expref{Definition}{def:SKE-generic}) satisfies $\QUF_t$.
\end{corollary}

\begin{proof}(Sketch.)
If there exists a QPT $\algo A$ which can break $\QUF_t$ for $\Pi^f$ using $t$-many queries, then it also breaks $\Pi^\algo F$ where $\algo F$ is a random function. If not, we construct an oracle distinguisher for $\algo O = f$ versus $\algo O = \algo F$ which simulates $\algo A$ in one of the two games (each with probability $1/2$) and outputs a bit which is biased depending on $\algo O$. Note that we only need $t$ queries to do this, since we only run one of the games (and not both). It then remains to invoke \expref{Theorem}{thm:QAErandom}, and observe that \expref{Theorem}{thm:QAEUF} holds in the case of a bounded number of queries.
\qed
\end{proof}

\subsubsection{Separations.}

While \QAE implies \QINDCCAA according to \expref{Theorem}{thm:QAECCA2}, the converse does not hold. 
In fact, consider any \QAE secure scheme and modify the decryption function by replacing the reject symbol by a fixed plaintext, e.g. the all zero state. Such a scheme is certainly still \QINDCCAA secure, as any adversary against it can be used against the original scheme by simulating the modified one. The modified scheme is, however, manifestly not \QAE as it never outputs $\bot$. The same reasoning works for \QUF in place of \QAE.
\begin{proposition}\label{prop:sepQINDCCAA-QUF}
	$\QINDCCAA\not\Rightarrow\QUF$, and therefore  $\QINDCCAA\not\Rightarrow\QAE$.
\end{proposition}

Finally, we turn to the relationship of \QAE and \QUF, and propose a separation as follows. Let $\Pi$ be a scheme that fulfills \cQCA (\expref{Definition}{def:cqCauth}) for trivial register $\tilde B$, but can be broken using an efficient attack with nontrivial $\tilde B$. For any PRF $f$, $\Pi^f$ is clearly \QUF, as the security notion ignores side information. It can however not be \QAE, as \QAE implies $\cQCA$.

\section{Discussion}

In this work, we presented four new security notions for symmetric key quantum encryption: \QCA, \QUF, \QINDCCAA and \QAE. While we have already made significant progress on understanding these notions, a number of open questions remain. A few are as follows. Does an encryption scheme as discussed below \expref{Proposition}{prop:sepQINDCCAA-QUF} exist, proving $\QUF\not\Rightarrow \QAE$? If so, does \QUF imply \QINDCCAA or \QINDCCA? Classically, unforgeability and \INDCCAA imply \CAE;  does this hold quantumly as well? Finally, is there a scheme that satisfies \QINDCCAA but cannot be upgraded to \QAE by simply modifying the decryption function?

\section{Acknowledgements}

The authors would like to thank Anne Broadbent, Fr\'ed\'eric Dupuis, Yfke Dulek, Alex Russell, Christian Schaffner, and Fang Song for insightful discussions about the problems solved in this work. The authors are indebted to Christopher Portmann who discovered an error in an earlier version of this paper.
Part of this work was done while T.G. was supported by the TU Darmstadt. Part of this work was done while G.A. and C.M. were at QMATH, University of Copenhagen. 
Part of this work was sponsored by the COST CryptoAction IC1306. T.G. acknowledges financial support from the European Commission’s PERCY grant (agreement 321310).  
G.A. and C.M. acknowledge financial support from the European Research Council (ERC Grant Agreement no 337603), the Danish Council for Independent Research (Sapere Aude) and VILLUM FONDEN via the QMATH Centre of Excellence (Grant No. 10059). 
This work is part of the research programme "Cryptography in the Quantum Age" with project number 639.022.519, which is financed by the Netherlands Organisation for Scientific Research (NWO).

\bibliographystyle{abbrv}


\end{document}